\documentclass[11pt]{article}

\usepackage{amssymb,amsmath,amsthm}
\usepackage{graphicx, color}
\usepackage[arrow, matrix, curve]{xy}
\usepackage[top=1in, bottom=1in, left=1in, right=1in]{geometry}

\usepackage{natbib}       
\usepackage{graphicx}                       
\usepackage{subfigure}
\usepackage{enumitem}                       
\usepackage{hyperref}
\usepackage[mathscr]{eucal}                 
\usepackage{dsfont}												
\usepackage{pstricks}
\usepackage{rotating}
\usepackage{lscape}
\usepackage{mathtools}                      
\mathtoolsset{showonlyrefs=true}            
\usepackage{fixltx2e,amsmath}               
\MakeRobust{\eqref}

\usepackage{mathdots}
\allowdisplaybreaks                         
\theoremstyle{plain}
\newtheorem{theorem}{Theorem}
\numberwithin{theorem}{section}
\newtheorem{lemma}[theorem]{Lemma}          
\newtheorem{proposition}[theorem]{Proposition}

\theoremstyle{definition}
\newtheorem{definition}[theorem]{Definition}

\newtheorem{remark}[theorem]{Remark}
\newtheorem{assumption}[theorem]{Assumption}

\renewcommand{\>}{\rangle}
\renewcommand{\(}{\left(}
\renewcommand{\)}{\right)}
\renewcommand{\[}{\left[}

\newcommand\Eb{\mathds{E}}
\newcommand\Pb{\mathds{P}}

\newcommand\Rb{\mathds{R}}

\newcommand\Ac{\mathscr{A}}

\newcommand\Bc{\mathscr{B}}
\newcommand\Cc{\mathscr{C}}

\newcommand\Gc{\mathscr{G}}
\newcommand\Hc{\mathscr{H}}

\newcommand\Mc{\mathscr{M}}
\newcommand\Lc{\mathscr{L}}

\newcommand\Nc{\mathscr{N}}
\newcommand\Oc{\mathcal{O}}

\newcommand\Qc{\mathscr{Q}}

\newcommand\Yc{\mathscr{Y}}

\newcommand\eps{a}

\newcommand\sig{\sigma}

\newcommand\Lam{\Lambda}

\newcommand\gam{\gamma}

\newcommand\lam{\lambda}
\newcommand\del{\delta}

\newcommand\xb{\bar{x}}
\newcommand\yb{\bar{y}}

\renewcommand\d{\partial}

\newcommand\dd{\mathrm{d}}
\newcommand\ee{\mathrm{e}}

\newcommand{\pa}{\partial}
\newcommand{\un}[1]{u^{(#1)}}
\newcommand{\qn}[1]{q^{(#1)}}

\newcommand{\Vn}[1]{V^{(#1)}}
\newcommand{\Vz}{V^{(0)}}
\newcommand{\Vo}{V^{(1)}}
\newcommand{\Vnb}[1]{\bar V^{(#1)}}

\newcommand{\Rz}{R^{(0)}}
\newcommand{\D}{\mathcal{D}}
\newcommand{\half}{\tfrac{1}{2}}
\def \eps {{a}}
\newcommand{\Acb}{\widehat{\Ac}}

\newcommand\Lch{\widehat{\Lc}}
\newcommand\Gch{\widehat{\Gc}}
\newcommand\Ych{\widehat{\Yc}}
\newcommand{\psib}{\bar{\psi}}

\begin{document}

\title{Portfolio Optimization under Local-Stochastic Volatility: Coefficient Taylor Series Approximations \& Implied Sharpe Ratio}

\author{
Matthew Lorig
\thanks{Department of Applied Mathematics, University of Washington, Seattle, WA, USA.
Work partially supported by NSF grant DMS-0739195.}
\and
Ronnie Sircar
\thanks{ORFE Department, Princeton University, Princeton, NJ, USA.
Work partially supported by NSF grant DMS-1211906.}
}

\date{\today}

\maketitle

\begin{abstract}
We study the finite horizon Merton portfolio optimization problem in a general local-stochastic volatility setting.  Using model coefficient expansion techniques, we derive approximations for the both the value function and the optimal investment strategy.  We also analyze the `implied Sharpe ratio' and derive a series  approximation for this quantity.
The zeroth-order approximation of the value function and optimal investment strategy correspond to those obtained by \cite{merton1969lifetime} when the risky asset follows a geometric Brownian motion.  
The first-order correction of the value function can, for general utility functions, be expressed as a differential operator acting on the zeroth-order term.
For power utility functions, higher order terms can also be computed as a differential operator acting on the zeroth-order term. We give a rigorous accuracy bound for the higher order approximations in this case in pure stochastic volatility models.   A number of examples are provided in order to demonstrate numerically the accuracy of our approximations.  
\end{abstract}


\section{Introduction}
The continuous time portfolio optimization problem was first studied by \cite{merton1969lifetime}, where 
he considers a market that contains a riskless bond, which grows at a fixed deterministic rate, and multiple risky assets, each of which is modeled as a geometric Brownian motion with constant drift and constant volatility.
In this setting, Merton obtains an explicit expression for the value function and optimal investment strategy of an investor who wishes to maximize expected utility when the utility function has certain specific forms.
However, much empirical evidence suggests that volatility is stochastic and is driven by both local and auxiliary factors, and so 
it is natural to ask how an investor would change his investment strategy in the presence of stochastic volatility.

There have been a number of studies in this direction, a few of which, we now mention.
\cite{dries2005} studies the finite horizon optimal investment problem in a CEV local volatility model.
\cite{chacko-viceira-2005} examine the infinite-horizon optimal investment problem in a Heston-like stochastic volatility model.
While both studies provide an explicit expression for an investor's value function and optimal investment strategy, the results are specific to the models studied in these two papers and for power utility functions. 

Approximation methods, which have been extensively used for option pricing and related problems, have been adapted for the portfolio selection problem, allowing for a wider class of volatility models and utility functions. The Merton problem for power utilities under {\em fast mean-reverting} stochastic volatility
was analyzed by asymptotic methods in \cite[Chapter 10]{fps-book}, and the related 
 partial hedging stochastic control problem in \citet{mattias1,montreal} using asymptotic analysis for the dual problem. More recently, \cite{fouque-sircar-zariphopoulou-2012} consider a general class of multiscale stochastic volatility models and general utility functions.  Here, volatility is driven by one fast-varying and one slow-varying factor.  The separation of time-scales allows the authors to obtain explicit approximations for the investor's value function and optimal control, by combining singular and regular perturbation methods on the primal problem.  These methods were previously developed to obtain explicit price approximations for various financial derivatives, as described in the book \cite{fpss}.

Here we study the Merton problem in a general local-stochastic volatility (LSV) setting. The LSV setting encompasses both local volatility models (e.g., CEV and quadratic) and stochastic volatility models (e.g., Heston and Hull-White) as well as models that combine local and auxiliary factors of volatility (e.g., SABR and $\lambda$-SABR).  As explicit expressions for the value function and optimal investment strategy are not available in this very general setting, we focus on obtaining approximations for these quantities.  Specifically, we will obtain an approximation for the solution of a nonlinear Hamilton-Jacobi-Bellman partial differential equation (HJB PDE) by expanding the PDE coefficients in a Taylor series.  The Taylor series expansion method was initially developed in \cite{pagliarani2011analytical} to solve linear pricing PDEs under local volatility models, and is closely related to the classical parametrix method (see, for instance, \cite{corielli2010parametrix} for applications in finance).  The method was later extended in \cite{lorig-pagliarani-pascucci-2} to include more general polynomial expansions and to handle multidimensional diffusions.  Additionally, the technique has been applied to models with jumps; see \cite{pascucci}, \cite{lorig-pagliarani-pascucci-1} and \cite{lorig-pagliarani-pascucci-5}.  We remark that the PDEs that arise in no-arbitrage pricing theory are linear, whereas the HJB PDE we consider here is fully nonlinear.

The rest of the paper proceeds as follows. In Section \ref{sec:merton}, we introduce a general class of local-stochastic volatility models, define a representative investor's value function and write the associated HJB PDE. Section \ref{sec:asymptotics} presents the first order approximation and formulas for the principal LSV correction to the value function and the optimal investment strategy. Motivated by the notion of Black-Scholes implied volatility,  we also develop the notion of implied Sharpe ratio, that provides a greater intuition about the resulting formulas, which are summarized in Section \ref{summ}. We discuss higher order terms in Section \ref{higher}, and show that power utilities are particularly amenable to obtaining explicit formulas for further terms in the approximation. In Section \ref{sec:examples}, we provide explicit results for power utility.  In particular, we derive rigorous error bounds for the value function in a stochastic volatility setting. In Section \ref{eq:examples2}, we provide two numerical examples, illustrating the accuracy and versatility of our approximation method.
Section \ref{sec:conclusion} concludes.

\section{Merton Problem under Local-Stochastic Volatility\label{sec:merton}}
We consider a local-stochastic volatility model for a risky asset $S$:
\begin{align}
\frac{\dd S_t}{S_t} & =  \tilde{\mu}(S_t,Y_t)\, \dd t + \tilde{\sig}(S_t,Y_t) \, \dd B_t^{(1)}\label{dSeqn}\\
\dd Y_t	&=	\tilde{c}(S_t,Y_t) \,\dd t + \tilde{\beta}(S_t,Y_t) \,\dd B_t^{(2)}, \nonumber
\end{align}
where $B^{(1)}$ and $B^{(2)}$ are standard Brownian motions under a probability measure $\Pb$ with correlation coefficient $\rho \in (-1,1)$: $\Eb\{\dd B_t^{(1)}\,\dd B_t^{(2)}\}=\rho\,\dd t$. The 
log price process $X=\log S$ is, by It\^o's formula, described by the following :
\begin{align}
\dd X_t	&=	b(X_t,Y_t)\,  \dd t + \sig(X_t,Y_t) \, \dd B_t^{(1)} \label{dXdY}\\
\dd Y_t	&=	c(X_t,Y_t) \,\dd t + \beta(X_t,Y_t) \,\dd B_t^{(2)} ,
\end{align}
where $\sig(X_t,Y_t)= \tilde{\sig}(e^{X_t},Y_t)$, and similarly $(\mu,c,\beta)$
from $(\tilde{\mu},\tilde{c},\tilde{\beta})$, and we have defined 
$$b(X_t,Y_t)= \mu(X_t,Y_t) - \frac{1}{2}\sig^2(X_t,Y_t).$$
The model coefficient functions $(\mu,\sigma,c,\beta)$ are smooth functions of $(x,y)$ and are such that the Markovian system \eqref{dXdY} admits a unique strong solution.

\subsection{Utility Maximization and HJB Equation}
We denote by $W$ the wealth process of an investor who invests $\pi_t$ units of currency in $S$ at time $t$ and invests $(W_t-\pi_t)$ units of currency in a riskless money market account.  For simplicity, we assume that the risk-free rate of interest is zero, and so the wealth process $W$ satisfies
\begin{align}
\dd W_t
	&=	\frac{\pi_t}{S_t} \,\dd S_t
	=		\pi_t \mu(X_t,Y_t)\, \dd t + \pi_t \sig(X_t,Y_t)\, \dd B_t^{(1)} .
\end{align}
The investor acts to maximize the expected utility of portfolio value, or wealth, at a fixed finite time horizon $T$: $\Eb \, \{U(W_T)\}$, where $U:\Rb_+\to\Rb$ is a smooth, increasing and strictly concave utility function satisfying the ``usual conditions'' $U'(0^+)=\infty$ and $U'(\infty)=0$. 

We define the investor's \emph{value function} $V$ by
\begin{equation}
V(t,x,y,w):=	\sup_{\pi \in \Pi} \Eb\{U(W_T)\mid X_t=x,Y_t=y,W_t=w\} , 
\label{eq:v.def}
\end{equation}
where $\Pi$ is the set of \emph{admissible strategies} $\pi$, which are non-anticipating and satisfy 
\begin{equation}
\Eb\left\{\int_0^T \pi_t^2 \sig^2 (X_t,Y_t)\,\dd t\right\} < \infty .
\end{equation}
and $W_t\geq0$ a.s.

We assume that $V \in C^{1,2,2,2}([0,T]\times\Rb\times\Rb\times\Rb_+)$.  The Hamilton-Jacobi-Bellman partial differential equation (HJB-PDE) associated with the stochastic control problem \eqref{eq:v.def} is
\begin{align}
\(\frac{\pa}{\pa t} + \Ac\) V + \max_{\pi\in\Rb} \Ac^\pi V
	&=	0 ,  &
V(T,x,y,w)
	&=	U(w) , \label{eq:hjb.pde}
\end{align}
where the operators $\Ac$ and $\Ac^\pi$ are given by
\begin{align}
\Ac
	&=	 \tfrac{1}{2} \sig^2(x,y) \frac{\pa^2}{\pa x^2}+ \rho \sig(x,y) \beta(x,y) \frac{\pa^2}{\pa x\pa y}
			 + \tfrac{1}{2} \beta^2(x,y) \frac{\pa^2}{\pa y^2} +  b(x,y) \frac{\pa}{\pa x}  + c(x,y) \frac{\pa}{\pa y}, \label{Adef}\\
\Ac^\pi
	&=\frac{1}{2} \pi^2\sig^2(x,y) \frac{\pa^2}{\pa w^2}	+ \pi\(\sig^2(x,y) \frac{\pa^2}{\pa x\pa w}
	+\rho \sig(x,y) \beta(x,y)\frac{\pa^2}{\pa y\pa w} + \mu(x,y) \frac{\pa}{\pa w}\).
\end{align}
We refer to the books \cite{flemingsoner} and \cite{pham} for technical details. 

The optimal strategy $\pi^*=\mbox{argmax}_\pi \,\Ac^\pi V$ is given (in feedback form) by
\begin{equation}
\pi^* = -\frac{\left(\sig^2(x,y)V_{xw}
	+\rho \sig(x,y) \beta(x,y)V_{yw} + \mu(x,y)V_w\right)}{\sig^2(x,y)V_{ww}},\label{eq:pi}
\end{equation}
where subscripts indicate partial derivatives. 

Inserting the optimal strategy $\pi^*$ into the HJB-PDE \eqref{eq:hjb.pde} yields
\begin{equation}
\( \frac{\pa}{\pa t}+ \Ac \) V + \Nc (V) , \qquad
V(T,x,y,w)
	=	U(w) , \label{eq:hjb.3} 
\end{equation}
where $\Nc (V)$ is a nonlinear term, which is given by
\begin{equation}
\Nc (V) 	= -\frac{\left(\sig(x,y)V_{xw}
	+\rho  \beta(x,y)V_{yw} + \lambda(x,y)V_w\right)^2}{2V_{ww}}, \label{Ndef}
\end{equation} 
and we have introduced the \emph{Sharpe ratio}
\begin{equation}
\lam(x,y)=	\frac{\mu(x,y)}{\sig(x,y)} . 
\end{equation}

\subsection{Constant Parameter Merton Problem\label{Mertonsec}}
We review and introduce notation that will be used later for the {\em constant parameter} Merton problem, that is, when $\tilde\mu$ and $\tilde\sigma$ in \eqref{dSeqn} are constant, and so therefore $\mu$ and $\sigma$ are constant and the stock price $S$ follows the geometric Brownian motion 
$$ \frac{dS_t}{S_t} = \mu\,dt + \sigma\, dB^{(1)}_t. $$
Then the Merton value function $M(t,w;\lambda)$ for the investment problem for this stock, whose constant Sharpe ratio is $\lambda=\mu/\sigma$, is the unique smooth solution of the HJB PDE problem 
\begin{equation}
M_t -  \tfrac{1}{2} \lam^2  \frac{M_w^2}{M_{ww}} =0, \qquad M(T,w)=U(w),  \label{M.pde}
\end{equation}
on $t<T$ and $w>0$. Smoothness of $M$ given a smooth utility function $U$ (as assumed above), as well as differentiability of $M$ in $\lambda$ is easily established by the Legendre transform, which converts \eqref{M.pde} into a linear constant coefficient parabolic PDE problem for the dual. Regularity results for the latter problem are standard. 

It is convenient to introduce the Merton risk tolerance function
\begin{equation}
R(t,w;\lam) := -\frac{M_w}{M_{ww}}(t,w;\lam), \label{Rzdef}
\end{equation}
and the operator notation
\begin{equation}
\D_k := \(R(t,w;\lam)\)^k\frac{\pa^k}{\pa w^k}, \quad k=1,2,\cdots. \label{Dkdef}
\end{equation}

We recall also the Vega-Gamma relationship taken from \cite[Lemma 3.2]{fouque-sircar-zariphopoulou-2012}:
\begin{lemma}\label{vegagammalemma}
The Merton value function $M(t,w;\lambda)$ satisfies the ``Vega-Gamma'' relation
\begin{equation}
 \frac{\pa M}{\pa\lambda} = -(T-t)\lambda\D_2M, \label{vegagamma}
 \end{equation}
where $\D_2$ is defined in \eqref{Dkdef}.
\end{lemma}
Thus the derivative of the value function with respect to the Sharpe ratio (analogous to an option price's derivative with respect to volatility, its Vega) is proportional to its negative ``second derivative'' $\D_2M$ (which is analogous to the option price's second derivative with respect to the stock price, its Gamma).
This result will be used repeatedly in deriving the implied Sharpe ratio in Section \ref{ISR} and the approximation to the optimal portfolio in Section \ref{optapprox}.

\section{Coefficient Expansion \& First Order Approximation}
\label{sec:asymptotics}
For general $\{c,\beta,\mu,\sig,\lam\}$  and $U$, there is no closed form solution of \eqref{eq:hjb.3}.  The goal of this section is to find series approximations for the value function 
$$V= \Vn{0} + \Vn{1} + \Vn{2}+\cdots,$$ 
and the optimal investment strategy 
$$\pi^* = \pi_0^* + \pi_1^* + \pi_2^*+\cdots,$$ 
using model coefficient (Taylor series) expansions.  This approach 
is developed for the linear European option pricing problem in a general LSV setting in \cite{lorig-pagliarani-pascucci-2}, where explicit approximations for option prices and implied volatilities are obtained by expanding the coefficients of the underlying diffusion as a Taylor series.  Note that, here the HJB-PDE \eqref{eq:hjb.3} is fully nonlinear. Our first order approximation formulas are summarized in Section \ref{summ}.

\subsection{Coefficient Polynomial Expansions}
We begin by fixing an point $(\xb,\yb) \in \Rb^2$.  For any function $\chi$ that is analytic in a neighborhood of $(\xb,\yb)$, we define the following family of functions indexed by $\eps\in[0,1]$:
\begin{equation}
\chi^\eps(x,y):=	\sum_{n=0}^\infty \eps^n \chi_n(x,y) , \label{eq:chi.expand} 
\end{equation}
where
\begin{equation}
\chi_n(x,y)
	:=	\sum_{k=0}^n \chi_{n-k,k} \cdot (x-\xb)^{n-k}(y-\yb)^k , \qquad
\chi_{n-k,k}
	:=	\frac{1}{(n-k)!k!} \d_x^{n-k}\d_y^k \chi(\xb,\yb), \label{eq:chi.n}
\end{equation}
and we note that $\chi_0=\chi_{0,0}=\chi(\xb,\yb)$ is a constant.
Observe that $\chi^\eps |_{\eps=1}$ is the Taylor series of $\chi$ about the point $(\xb,\yb)$. Here, $\eps$ is an {\em accounting parameter} that will be used to identify successive terms of our approximation.

In the PDE \eqref{eq:hjb.3}, we will replace each of the coefficient functions
$$ \chi\in \{ \mu, c, \sig^2, \beta^2, \lam^2, \sig \beta, \beta \lam \}$$
by $\chi^\eps$, for some $\eps,(\xb,\yb)$, and then use the series expansion \eqref{eq:chi.expand} for $\chi^\eps$. 
Another way of saying this is that we assume the coefficients are of the form
$$ \chi\Big(\xb+\eps(x-\xb), \yb+\eps(y-\yb)\Big),$$
whose exact Taylor series is given by \eqref{eq:chi.expand}, and we are interested in the case when $\eps=1$. 

Consider now the following family of HJB-PDE problems
\begin{equation}
	\(\frac{\pa}{\pa t} + \Ac^\eps\)V^\eps + \Nc^\eps(V^\eps)=0,\qquad V^\eps(T,x,y,w)
	=	U(w), \label{eq:u.eps.pde}
\end{equation}
where, for $\eps\in [0,1]$,  $\Ac^\eps$ and $\Nc^\eps(\cdot)$ are obtained from $\Ac$ and $\Nc(\cdot)$ in \eqref{Adef} and \eqref{Ndef} by making the change
\begin{align}
\{ \mu, c, \sig^2, \beta^2, \lam^2, \sig \beta, \beta \lam \}
	&\xmapsto{    }\{ \mu^\eps, c^\eps, (\sig^2)^\eps, (\beta^2)^\eps, (\lam^2)^\eps, (\sig \beta)^\eps, (\beta \lam)^\eps \} .
\end{align}
The linear operator in the PDE \eqref{eq:u.eps.pde} can therefore be written as
$$ \Ac^\eps=\sum_{n=0}^\infty\eps^n\Ac_n,$$ where we define
\begin{align}
\Ac_n
	&:=	( \tfrac{1}{2} \sig^2 )_n(x,y)\frac{\pa^2}{\pa x^2} 			+ (\rho \sig \beta )_n(x,y) \frac{\pa^2}{\pa x\pa y}
 + ( \tfrac{1}{2} \beta^2 )_n(x,y)\frac{\pa^2}{\pa y^2} 	
+b_n(x,y) \frac{\pa}{\pa x}+ c_n(x,y)\frac{\pa}{\pa y} ,
\label{eq:An}
\end{align}
and the expansion of the nonlinear term is a more involved computation.

We construct a series approximation for the function $V^\eps$ as a power series in $\eps$:
\begin{align}
V^\eps(t,x,y,w)
	&=	\sum_{n=0}^\infty \eps^n V^{(n)}(t,x,y,w) . \label{eq:u.expand} 
\end{align}
Note that the functions $V^{(n)}$ are not constrained to be polynomials in $(x,y,w)$, and in general they will not be. 
Our approximate solution to \eqref{eq:hjb.3}, which is the problem of interest, will then follow by setting $\eps=1$.

\subsection{Zeroth \& First Order Approximations}
We insert \eqref{eq:u.expand} into \eqref{eq:u.eps.pde} and collect terms of like powers of $\eps$.  At lowest order we obtain
\begin{equation}
\(\frac{\pa}{\pa t} + \Ac_0\) \Vn{0} -\frac{\left(\sig_0\Vz_{xw}
	+\rho  \beta_0\Vz_{yw} + \lambda_0\Vz_w\right)^2}{2\Vz_{ww}}=0, \qquad \Vz(T,x,y,w)=U(w),  
\label{eq:u0}
\end{equation}
where the linear operator $\Ac_0$, found from \eqref{eq:An}, has constant coefficients:
\begin{equation}
\Ac_0
	=	\tfrac{1}{2} \sig^2_0\frac{\pa^2}{\pa x^2}+ \rho \sig_0 \beta_0\frac{\pa^2}{\pa x\pa y}
 + \tfrac{1}{2} \beta^2_0\frac{\pa^2}{\pa y^2} +b_0\frac{\pa}{\pa x}
 + c_0\frac{\pa}{\pa y}. \label{A0def}
 \end{equation}
  
As a consequence, the solution of \eqref{eq:u0} is independent of $x$ and $y$: $\Vz=\Vz(t,w)$, and we have
\begin{equation}
\Vz_t -  \tfrac{1}{2} \lam^2 _0  \frac{\(\Vz_w\)^2}{\Vz_{ww}} =0, \qquad \Vz(T,w)=U(w).  \label{eq:u0.pde}
\end{equation}
We observe that \eqref{eq:u0.pde} is the same as the PDE problem \eqref{M.pde} that arises when solving the Merton problem assuming the underlying stock has a constant drift $\mu_0=\mu(\xb,\yb)$, diffusion coefficient $\sig_0=\sig(\xb,\yb)$ and so constant Sharpe ratio $\lam_0 = \lam(\xb,\yb) = \mu(\xb,\yb)/\sig(\xb,\yb)$. Therefore, we have
$$ \Vz(t,w) = M(t,w;\lam_0). $$
The PDE \eqref{eq:u0.pde} can be solved either analytically (for certain utility functions $U$), or numerically.  

Recall the definition of the risk tolerance function in \eqref{Rzdef} and the operators $\D_k$ in \eqref{Dkdef}, where now we take in those formulas the Sharpe ratio $\lam_0$:
\begin{equation}
R(t,w;\lam_0) = -\frac{\Vn{0}_w}{\Vn{0}_{ww}}(t,w;\lam),\qquad \D_k = \(R(t,w;\lam_0)\)^k\frac{\pa^k}{\pa w^k}, \quad k=1,2,\cdots. \label{Dkdef2}
\end{equation}
Proceeding to the order $\eps$ terms in \eqref{eq:u.eps.pde}, we obtain
\begin{equation}
\(\frac{\pa}{\pa t} + \Ac_0\) \Vn{1} + \half\lambda^2_0\D_2\Vo + \lambda_0^2\D_1\Vo + \rho\beta_0\lambda_0\D_1\frac{\pa}{\pa y}\Vo + \mu_0\D_1\frac{\pa}{\pa x}\Vo =-(\tfrac{1}{2}\lam^2)_1\D_1\Vz.
\end{equation}
We can re-write this more compactly as
\begin{equation}
\(\frac{\pa}{\pa t} + \Ac_0+ \Bc_0\) \Vn{1} + H_1 =0, \qquad \Vo(T,x,y,w)=0,\label{Voeqn}
\end{equation}
where the linear operator $\Bc_0$ and the source term $H_1$ are given by
\begin{align}
\Bc_0
	&=	\tfrac{1}{2}\lam^2_0 \D_2+ \lam^2_0\D_1
			+ \rho \beta_0 \lam_0  \D_1\frac{\pa}{\pa y}
			+ \mu_0 \D_1\frac{\pa}{\pa x} , \label{eq:B0}\\
H_1(t,x,y,w) & = 	(\tfrac{1}{2}\lam^2)_1(x,y)\D_1\Vz(t,w).	\label{eq:H1}	
\end{align}

\subsection{Transformation to Constant Coefficient PDEs}
Next, we apply a change of variable such that $\Vo$ can be found by solving a linear PDE with constant coefficients. 
We begin with the following lemma.
\begin{lemma}
Let $\Vz$ be the solution of \eqref{eq:u0.pde} and let $\Ac_0$ and $\Bc_0$ be as given in \eqref{eq:An} and \eqref{eq:H1}, respectively.  Then $\Vz$ satisfies the following PDE problem
\begin{align}
\(\frac{\pa}{\pa t} + \Ac_0+ \Bc_0\)\Vz=0, \qquad \Vz(T,w)=	U(w) . \label{eq:u0.pde.3}
\end{align}
\end{lemma}
\begin{proof}
This follows directly from observing that the nonlinear term in \eqref{eq:u0.pde} can be written
\begin{equation}
 \frac{\(\Vz_w\)^2}{\Vz_{ww}}=\left(-\frac{\Vz_w}{\Vz_{ww}}\right)^2\Vz_{ww} = \D_2\Vz, \quad\mbox{ or }\quad
  \frac{\(\Vz_w\)^2}{\Vz_{ww}}=-\left(-\frac{\Vz_w}{\Vz_{ww}}\right)\Vz_{w} = -\D_1\Vz.
  \label{driftdiff}
 \end{equation}
Therefore, from \eqref{eq:u0.pde}, we have
\begin{equation} 
\(\frac{\pa}{\pa t}+\tfrac{1}{2}\lam^2_0 \D_2+ \lam^2_0\D_1\)\Vz=0, \label{u0main}
\end{equation}
and \eqref{eq:u0.pde.3} follows from the fact that $\Vz$ does not depend on $(x,y)$, while $\Ac_0$ and the last two terms in the expression \eqref{eq:B0} for $\Bc_0$ take derivatives in those variables.
\end{proof}

Next, it will be helpful to introduce the following change of variables.
\begin{definition}
\label{def:q}
We define the co-ordinate $z$ by the transformation 
\begin{align}
z(t,w)
	&=		- \log  \Vz_w(t,w) + \tfrac{1}{2}\lam^2_0 (T-t). \label{eq:z}
\end{align}
\end{definition}
We have the following change of variables formula, as used also in \cite[Section 2.3.2]{fouque-sircar-zariphopoulou-2012}.
\begin{lemma}\label{princtransf}
For a smooth function $\widehat{V}(t,x,y,w)$, define $q(t,x,y,z)$ by
$$ \widehat{V}(t,x,y,w)=q(t,x,y,z(t,w)). $$
Then we have
\begin{equation} 
\(\frac{\pa}{\pa t} + \Ac_0+ \Bc_0\)\widehat{V} = \(\frac{\pa}{\pa t} + \Ac_0+ \Cc_0\)q,
\label{transf}
\end{equation}
where the operator $\Cc_0$ is given by
\begin{align}
\Cc_0 	&=	\tfrac{1}{2}\lam^2_0\frac{\pa^2}{\pa z^2}
			+ \rho \beta_0 \lam_0 \frac{\pa^2}{\pa y\pa z}
			+ \mu_0 \frac{\pa^2}{\pa x\pa z}. \label{eq:C0}
\end{align}
\end{lemma}
\begin{proof}
We shall use the shorthand $\Rz(t,w)=R(t,w;\lam_0)$. 
From \eqref{eq:z}, we have that $z_w=1/\Rz$, and so, differentiating \eqref{eq:q}, we find
\begin{equation}
\widehat{V}_t=q_t-\(\frac{\Vz_{tw}}{\Vz_w}+\half\lam_0^2\)q_z,\qquad \D_1\widehat{V}=q_z, \qquad \D_2\widehat{V}=q_{zz}-\Rz_wq_z.
\end{equation}
Then, using the first expression in \eqref{driftdiff} to write the PDE \eqref{eq:u0.pde} for $\Vz$ as $\Vz_t=\tfrac{1}{2}\lam^2_0\D_2\Vz$, and differentiating this with respect to $w$ gives
$$ \Vz_{tw}=\half\lam_0^2\(\Rz\)^2\Vz_{www} + \lam_0^2\Rz\Rz_w\Vz_{ww}. $$
But from $\Rz\Vz_{ww} = -\Vz_w$, we have $(\Rz)^2\Vz_{www}=(\Rz_w+1)\Vz_w$, and so
$$ \Vz_{tw} = \half\lam_0^2(\Rz_w+1)\Vz_w - \lam_0^2\Rz_w\Vz_w, $$
which gives that
\begin{equation}
\frac{\Vz_{tw}}{\Vz_w} =-\half\lam_0^2(\Rz_w-1).
\end{equation}
Therefore, we have
$$ \(\frac{\pa}{\pa t}+\tfrac{1}{2}\lam^2_0 \D_2+ \lam^2_0\D_1\)\widehat{V}=q_t+\(\half\lam_0^2(\Rz_w-1)-\half\lam_0^2\)q_z + \tfrac{1}{2}\lam^2_0\(q_{zz}-\Rz_wq_z\) + \lam_0^2q_z,$$
which establishes that
\begin{equation}
\(\frac{\pa}{\pa t}+\tfrac{1}{2}\lam^2_0 \D_2+ \lam^2_0\D_1\)\widehat{V} = 
\(\frac{\pa}{\pa t} + \tfrac{1}{2}\lam^2_0 \frac{\pa^2}{\pa z^2}\)q. \label{intm}
\end{equation}
More directly, we have
$$ \left(\rho\beta_0\lambda_0\D_1\frac{\pa}{\pa y} + \mu_0\D_1\frac{\pa}{\pa x}\right)\widehat{V} = \left(\rho \beta_0 \lam_0 \frac{\pa^2}{\pa y\pa z}+ \mu_0 \frac{\pa^2}{\pa x\pa z}\right)q,$$
which, combined with \eqref{intm}, leads to \eqref{transf}.
\end{proof}
We define $\qn{0}$ by $\Vn{0}(t,w)=	\qn{0}(t,z(t,w))$.  
Then the PDE \eqref{eq:u0.pde.3} for $\Vn{0}$ is transformed to the (constant coefficient) backward heat equation for $\qn{0}(t,z)$:
\begin{equation}
\(\frac{\pa}{\pa t} + \tfrac{1}{2}\lam^2_0\frac{\pa^2}{\pa z^2}\)\qn{0}=0, \qquad \qn{0}(T,z)=U\((U')^{-1}(e^{-z})\),\label{q0PDE1}
\end{equation}
but of course the transformation \eqref{eq:z} depends on the solution $\Vz$ itself. Again, as $\qn{0}$ does not depend on $(x,y)$, while  $\Ac_0$ and the last two terms in the expression \eqref{eq:C0} for $\Cc_0$ take derivatives in those variables, we can write
\begin{equation}
\(\frac{\pa}{\pa t} + \Ac_0+ \Cc_0\)\qn{0}=0. \label{q0PDE}
\end{equation}

Now let $\qn{1}$ be defined from $\Vn{1}$ by 
\begin{align}
\Vn{1}(t,x,y,w)
	&=	\qn{1}(t,x,y,z(t,w)) , \label{eq:q}
\end{align}
using the transformation \eqref{eq:z}.
Then, using Lemma \ref{princtransf}, we see that the PDE \eqref{Voeqn} for $\Vn{1}$, which has $(t,w)$-dependent coefficients through the dependence of $\Bc_0$ in \eqref{eq:B0} on $R(t,w;\lam_0)$, is transformed to the {\em constant coefficient} equation for $\qn{1}$:
\begin{equation}
\(\frac{\pa}{\pa t} + \Ac_0+ \Cc_0\)\qn{1}+Q_1=0, \qquad \qn{1}(T,x,y,z)=0. \label{q1eqn}
\end{equation}
The source term is found from $H_1(t,x,y,w)=Q_1(t,x,y,z(t,w))$, where, from \eqref{eq:H1}, we have 
\begin{equation}
Q_1(t,x,y,z) = (\tfrac{1}{2}\lam^2)_1(x,y)\qn{0}_z. \label{Q1formula}
\end{equation}

\subsection{Explicit expression for $\Vn{1}$}
\label{sec:explicit}
In this section, we will show that 
$\Vo$, solution of \eqref{Voeqn}, can be written as a differential operator acting on $\Vz$. 
First, we look at the PDE problem
\begin{equation}
\Hc q+Q=0, \qquad q(T,x,y,z)=0, \label{qeqn}
\end{equation}
where $\Hc$ is the constant coefficient linear operator
\begin{equation}
\Hc = \frac{\pa}{\pa t} + \Ac_0+ \Cc_0.
\end{equation}
We also suppose that the source term $Q(t,x,y,z)$ is of the following special form:
\begin{equation}
Q(t,x,y,z) = \sum_{k,l,n}(T-t)^n(x-\xb)^k(y-\yb)^lv(t,x,y,z), \label{Qspec}
\end{equation}
where the sum is finite and $v$ is a solution of the homogeneous equation
\begin{equation}
\Hc v=0. \label{veqn}
\end{equation}

We first define (the commutator) $\Lc_{X}=[\Hc,(x-\xb)I]$ by 
\begin{equation}
 \Hc\((x-\xb)v\) = (x-\xb)\Hc v + \Lc_{X}v, \label{LXprop}
 \end{equation}
and so a direct calculation using the expressions \eqref{A0def} and \eqref{eq:C0} for $\Ac_0$ and $\Cc_0$ respectively shows that
\begin{equation}
\Lc_{X} = (\mu_0-\half\sigma_0^2)I + \sigma_0^2\frac{\pa}{\pa x} + \rho\sigma_0\beta_0\frac{\pa}{\pa y} + \mu_0\frac{\pa}{\pa z},\label{LXformula}
\end{equation}
where $I$ is the identity operator. 
Similarly defining (the commutator) $\Lc_{Y}=[\Hc,(y-\yb)I]$ by 
$$ \Hc\((y-\yb)v\) = (y-\yb)\Hc v + \Lc_{Y}v, $$
leads to 
\begin{align}
\Lc_{Y} &= c_0I + \beta^2_0\frac{\pa}{\pa y} + \rho\sigma_0\beta_0\frac{\pa}{\pa x} + \rho\beta_0\lam_0\frac{\pa}{\pa z}, \label{LYformula} 
\end{align}

We next introduce the following operators indexed by $s\in[t,T]$:
\begin{equation}
\Mc_{X}(s) = (x-\xb)I + (s-t)\Lc_{X}, \qquad 
\Mc_{Y}(s) = (y-\yb)I + (s-t)\Lc_{Y}, 
\label{Mopdef}
\end{equation}
Then we have the following result by construction of these operators.
\begin{lemma}\label{HMlem}
Recall that $v$ solves the homogeneous equation \eqref{veqn}. Then 
\begin{equation} 
\Hc\Mc_{X}^k(s)\Mc_{Y}^l(s)v=0, \label{cascade}
\end{equation}
for integers $k,l$.
\end{lemma}
\begin{proof}
We first calculate
$$ \Hc\Mc_{X}v = \Mc_{X}\Hc v + \Lc_{X}v - \Lc_{X}v +(s-t)\Hc\Lc_{X}v=(s-t)\Hc\Lc_{X}v, $$
where we have used \eqref{LXprop}. But since $\Hc$ and $\Lc_{X}$ are constant coefficient operators which commute, we have $\Hc\Lc_{X}v=\Lc_{X}\Hc v=0$ using \eqref{veqn}.
Therefore, given a solution $v$ of the homogeneous equation, $\Mc_{X}v$ also solves the homogeneous equation, namely $\Hc\(\Mc_{X}v\)=0$. Iterating we have that $\Hc\Mc_{X}^kv=0$ for integers $k$. Similarly $\Hc\Mc_{Y}^lv=0$ for integers $l$, and so the result \eqref{cascade} follows. 
\end{proof}

From this we can exploit the special structure of the source $Q$ to obtain the following formula.
\begin{proposition}\label{qformulaprop}
The solution to \eqref{qeqn} where the source $Q$ is of the form \eqref{Qspec} is given by
\begin{equation}
q(t,x,y,z) = \sum_{k,l,n}\int_t^T(T-s)^n\Mc_X^k(s)\Mc_Y^l(s)v(t,x,y,z)\, ds. \label{qformula}
\end{equation}
\end{proposition}
\begin{proof}
Due to the linearity of the problem, it suffices to consider a single term of the polynomial: 
$$ Q(t,x,y,z) = (T-t)^n(x-\xb)^k(y-\yb)^lv(t,x,y,z). $$
Then, we check that the solution is given by
$$ q(t,x,y,z) =\int_t^T(T-s)^n\Mc_X^k(s)\Mc_Y^l(s)v(t,x,y,z)\, ds $$
by computing
\begin{align}
\Hc q &= -(T-t)^n\Mc_X^k(t)\Mc_Y^l(t)v(t,x,y,z)  + \int_t^T(T-s)^n\Hc\Mc_X^k(s)\Mc_Y^l(s)v(t,x,y,z)\, ds\\
&=-(T-t)^n(x-\xb)^k(y-\yb)^lv(t,x,y,z)\\
&=-Q,
\end{align}
using Lemma \ref{HMlem} for the second term. 
The formula \eqref{qformula} in the general polynomial case follows, and clearly the zero terminal condition is satisfied by \eqref{qformula}. 
\end{proof}
We can now solve for the first correction in the series expansion.
\begin{proposition}\label{q1prop}
The solution to \eqref{q1eqn} is given by
\begin{equation} 
\qn{1}(t,x,y,z) = (T-t)\lam_0A(t, x, y)\,\qn{0}_z(t,z)+ \half(T-t)^2\lam_0B\,\qn{0}_{zz}(t,z),\label{q1formula}
\end{equation}
where
\begin{align}
A(t, x, y) &=  \lambda_{1,0}\left[(x-\bar x)+\half(T-t)(\mu_0-\half\sigma_0^2)\right]
+\lambda_{0,1}\left[(y-\bar y)+\half(T-t)c_0\right], \nonumber\\
 B&=\lambda_{1,0}\mu_0 + \lambda_{0,1}\rho\beta_0\lam_0. \label{ABdef}
\end{align}
\end{proposition}
\begin{proof}
We observe that since $\qn{0}$ satisfies the homogeneous PDE $\Hc\qn{0}=0$ from \eqref{q0PDE}, so does $\qn{0}_z$, which follows from differentiating the constant coefficient PDE for $\qn{0}$. Then applying  Proposition \ref{qformulaprop} with $v=\qn{0}_z$, and $n=0, (k,l)\in\{(1,0), (0,1)\}$ and substituting the definitions \eqref{Mopdef} for $\Mc_{X}$ and $\Mc_{Y}$ leads to 
\begin{align*} 
\qn{1}(t,x,y,z) &= \left[(\half\lambda^2)_{1,0}\left((T-t)(x-\bar x)+\half(T-t)^2\Lc_X\right)\right.\\
&\left. +(\half\lambda^2)_{0,1}\left((T-t)(y-\bar y)+\half(T-t)^2\Lc_Y\right)\right]\qn{0}_z(t,z).
\end{align*}
Finally, substituting for $\Lc_{X}$ and $\Lc_{Y}$ from \eqref{LXformula}
and \eqref{LYformula} and using that $\qn{0}$ does not depend on $(x,y)$ leads to \eqref{q1formula}. 
\end{proof}

In the original variables, this leads to 
\begin{equation}
\Vn{1}(t,x,y,w) = (T-t)\lam_0A(t, x, y)\,\D_1\Vn{0}(t,w)+ \half(T-t)^2\lam_0B\,\D_1^2\Vn{0}(t,w). \label{V1formula}
\end{equation}

\subsection{Implied Sharpe Ratio\label{ISR}}
In an analogy to implied volatility, for a fixed maturity $T$ and utility function $U$, one can define the \emph{Merton implied Sharpe ratio}\footnote{The authors thank Jean-Pierre Fouque for a number of fruitful discussions,
from which the concept of the Merton implied Sharpe ratio arose.} corresponding to value function $M(t,w;\lambda)$ of Section \ref{Mertonsec} as the unique positive solution $\Lam^\eps(t,x,y,w)$ of
\begin{align}
V^\eps(t,x,y,w)
	&=	M(t,w;\Lam^\eps) . \label{eq:implied.sharpe}
\end{align}
The existence and uniqueness of the implied Sharpe ratio follows from the fact that (i) the function $M$ satisfies $M(t,w) \geq U(w)$, since an investor with initial wealth $w$ can always obtain a terminal utility $U(w)$ by investing all of his money in the riskless bank account, and (ii) the function $M$ is strictly increasing in $\Lam$.  Since a higher implied Sharpe ratio is indicative of a better investment opportunity, we are interested to know how local stochastic volatility model parameters $\{c,\beta,\mu,\sig,\rho\}$ affect the implied Sharpe ratio.

Using our first order approximation $V^\eps\approx\Vn{0}+\eps\Vn{1}$, we look for a corresponding series approximation of the implied Sharpe ratio as
$$\Lam^\eps = \Lam^{(0)} + \eps  \Lam^{(1)} + \cdots. $$
Then, expanding 
$$ M(t,w;\Lam) = M(t,w;\Lam^{(0)}) + \eps \Lam^{(1)}M_\lam(t,w;\Lam^{(0)}) + \cdots, $$
and comparing with the expansion 
$$ V^\eps(t,x,y,w)=M(t,w;\lam_0) + \eps \Vn{1}(t,x,y,w) +\cdots $$
yields $\Lam^{(0)}=\lam_0$ and
\begin{equation}
\Lam^{(1)}=\frac{\Vn{1}(t,x,y,w)}{M_\lam(t,w;\lam^{(0)})}. \label{Lam1}
\end{equation}
Next, from Lemma \ref{vegagammalemma}, we have 
$$M_\lam(t,w;\lam^{(0)})=-(T-t)\lam_0\D_2M(t,w;\lam^{(0)}) = -(T-t)\lam_0\D_2\Vn{0}(t,w)=(T-t)\lam_0\D_1\Vn{0}(t,w). $$
Then, using the formula \eqref{V1formula} for $\Vn{1}$ in \eqref{Lam1} gives
$$ \Lam^{(1)}(t,x,y,w)=A(t, x, y) +\half(T-t)B\,\frac{\D_1^2\Vn{0}(t,w)}{\D_1\Vn{0}(t,w)}. $$
By computing
$$ \frac{\D_1^2\Vn{0}(t,w)}{\D_1\Vn{0}(t,w)}=R_w(t,w;\lam_0)-1, $$
we have
\begin{equation}
\Lam^\eps \approx \Lam^{(0)} + \eps  \Lam^{(1)} =\lam_0 + \eps\left[A(t, x, y) +\half(T-t)B\left(R_w(t,w;\lam_0)-1\right)\right], \label{Lamapprox}
\end{equation}
where $R$ is the Merton risk tolerance function
$$ R(t,w;\lam)=-\frac{M_w(t,w;\lam)}{M_{ww}(t,w;\lam)}. $$

\subsection{Optimal Portfolio\label{optapprox}}
From \eqref{eq:pi}, we have that the optimal strategy $\pi^{\eps,*}$ is given by
\begin{equation}
\pi^{\eps,*} =  -\frac{\mu^\eps(x,y)V^\eps_w}{(\sig^\eps)^2(x,y)V^\eps_{ww}} - \frac{\rho\beta^\eps(x,y)V^\eps_{yw}}{\sig^\eps(x,y)V^\eps_{ww}} - \frac{V^\eps_{xw}}{V^\eps_{ww}}. \label{epspolicy}
\end{equation}
It is convenient in deriving a compact form for our portfolio approximation to write our first order approximation
to the value function as the Merton value function evaluated at the first order series \eqref{Lamapprox} for the Sharpe ratio 
$$ V^\eps(t,x,y,w) \approx \bar V(t,x,t,w):=M(t,w;\lam_0 + \eps  \Lam^{(1)}(t,x,y,w)).$$
Then our approximate first order policy will be to substitute $\bar V$ for $V^\eps$ in \eqref{epspolicy}. 

We have
\begin{eqnarray*}
\bar V_w(t,x,y,w) & = & M_w\(t,w;\lam_0 + \eps  \Lam^{(1)}(t,x,y,w)\) + \eps M_\lam(t,w;\lam_0 )\half(T-t)BR_{ww}(t,w;\lam_0)+ \mathcal{O}(\eps^2),\\
\bar V_{ww}(t,x,y,w) & = & M_{ww}\(t,w;\lam_0 + \eps  \Lam^{(1)}(t,x,y,w)\) + 
\eps\half(T-t)B\left(R_{ww}(t,w;\lam_0)M_\lam(t,w;\lam_0)\right)_w+\mathcal{O}(\eps^2),\\
\bar V_{yw}(t,x,y,w) & = & \eps\lam_{0,1}M_{\lam w}(t,w;\lam_0) + \mathcal{O}(\eps^2),\\
\bar V_{xw}(t,x,y,w) & = & \eps\lam_{1,0}M_{\lam w}(t,w;\lam_0) + \mathcal{O}(\eps^2),
\end{eqnarray*}
where $\mathcal{O}(\eps^2)$ denotes series terms in powers of $\eps^2$ and higher.

Let us compute
\begin{eqnarray*}
- \frac{\bar V_w}{\bar V_{ww}} &=& -\frac{M_w\(t,w;\lam_0 + \eps  \Lam^{(1)}(t,x,y,w)\) + \eps M_\lam(t,w;\lam_0 )\half(T-t)BR_{ww}(t,w;\lam_0)}{M_{ww}\(t,w;\lam_0 + \eps  \Lam^{(1)}(t,x,y,w)\) + 
\eps\half(T-t)B\left(R_{ww}(t,w;\lam_0)M_\lam(t,w;\lam_0)\right)_w}\\
&=& R\(t,w;\lam_0 + \eps  \Lam^{(1)}(t,x,y,w)\)
-\eps\frac{M_\lam}{M_{ww}}\half(T-t)BR_{ww} + \eps\half(T-t)B\frac{M_w}{M_{ww}^2}(R_{ww}M_\lam)_w+ \mathcal{O}(\eps^2).\\
&=& R\(t,w;\lam_0 + \eps  \Lam^{(1)}(t,x,y,w)\)+\eps\half(T-t)^2B\lam_0R^2(R_{ww}+RR_{www}+(R_w-1)R_{ww})+ \mathcal{O}(\eps^2).
\end{eqnarray*}
Here we have used the following identities satisfied by the Merton value function $M(t,w;\lam)$ and its risk tolerance function $R(t,w;\lam)$:
\begin{align}
\frac{M_\lam}{M_{ww}} & =  -(T-t)\lam R^2,\label{id1}\\
\frac{M_w}{M_{ww}^2}M_{\lam} & =  (T-t)\lam R^3,\label{id2}\\
\frac{M_w}{M_{ww}^2}M_{\lam w} & =  (T-t)\lam R^2(R_w-1),\label{id3}
\end{align}
where \eqref{id1} comes from Lemma \ref{vegagammalemma}; \eqref{id2} comes from multiplying \eqref{id1} by $-R$; 
and in the last expression \eqref{id3}, we also use $R^2M_{www}=(R_w+1)M_w$.

Additionally, we compute 
\begin{eqnarray*}
- \frac{\bar V_{yw}}{\bar V_{ww}} & = & -\eps\lam_{0,1}\frac{M_{\lam w}}{M_{ww}}(t,w;\lam_0)
= \eps\lam_{0,1}(T-t)\lam_0R(t,w;\lam_0)\(R_w(t,w;\lam_0)-1\)+ \mathcal{O}(\eps^2), \\
- \frac{\bar V_{xw}}{\bar V_{ww}} & = & -\eps\lam_{1,0}\frac{M_{\lam w}}{M_{ww}}(t,w;\lam_0) =  \eps\lam_{1,0}(T-t)\lam_0R(t,w;\lam_0)\(R_w(t,w;\lam_0)-1\)+ \mathcal{O}(\eps^2).
\end{eqnarray*}

Therefore we have
\begin{align}
\pi^{\eps,*} & \approx  \frac{\mu^\eps(x,y)}{(\sig^\eps)^2(x,y)}\left\{R\(t,w;\lam_0 + \eps  \Lam^{(1)}(t,x,y,w)\)+\half\eps(T-t)^2B\lam_0R^2(RR_{www}+(R+R_w-1)R_{ww})\right\}\nonumber\\
& + \eps(T-t)\lam_0R(R_w-1)\left(\frac{\rho\beta^\eps(x,y)}{\sig^\eps(x,y)}\lam_{0,1}+\lam_{1,0}\right), \label{pistarapprox}
\end{align}
where $R$ without an argument denotes $R(t,w;\lam_0)$. One could substitute the first two terms of the polynomial expansion of the coefficients, but since they are assumed known, there is no loss in accuracy in using the full expressions. The first order approximate optimal strategy is written in terms of the risk tolerance function and its derivatives. 

\subsection{Summary\label{summ}}
We collect here the expressions for our first order approximation formulas, which follow from the prior calculations and setting the accounting parameter $a=1$.
\begin{itemize}
\item Our first order approximation to the value function $V(t,x,y,w)$ in \eqref{eq:v.def}, solution of the PDE problem \eqref{eq:hjb.3} is given by $V(t,x,y,w)\approx \bar V(t,x,y,w)$, where
\begin{align}
 \bar V(t,x,y,w) & = \Vz(t,w) + \Vo(t,x,y,w)\\
 &= M(t,w;\lam_0) + \((T-t)\lam_0A(t, x, y)\,\D_1+ \half(T-t)^2\lam_0B\,\D_1^2\)M(t,w;\lam_0), 
\end{align}
and $A$ and $B$ are given in \eqref{ABdef}.

\item The implied Sharpe ratio $\Lam=\Lam(t,x,y,w)$ defined by $V(t,x,y,w)=M(t,w;\Lam)$ is approximated to first order by
$\Lam\approx\bar\Lam$, where
\begin{align}
\bar\Lam(t,x,y,w) &=\Lam^{(0)} +   \Lam^{(1)}\\
&=\lam_0 + A(t, x, y) +\half(T-t)B\left(R_w(t,w;\lam_0)-1\right).\label{ISRsumm}
\end{align}

\item Our first order approximation to the optimal strategy $\pi^*(t,x,y,w)$ in \eqref{eq:pi} is given by $\pi^*\approx\bar\pi$, where 
\begin{align}
\bar\pi(t,x,y,w) &= \frac{\mu(x,y)}{\sig^2(x,y)}\left\{R\(t,w;\lam_0 +   \Lam^{(1)}(t,x,y,w)\)\right.\nonumber\\
&\qquad\qquad\quad+\half(T-t)^2B\lam_0R^2(RR_{www}+(R+R_w-1)R_{ww})\Big\}\nonumber\\
& + (T-t)\lam_0R(R_w-1)\left(\frac{\rho\beta(x,y)}{\sig(x,y)}\lam_{0,1}+\lam_{1,0}\right).\label{pistarapprox1}
\end{align}
This formula has principle term that is the classical Merton strategy $-\frac{\mu}{\sigma^2}R$, but here is updated to account for LSV by using the current $\mu(x,y)$ and $\sigma(x,y)$ values, and with the implied Sharpe ratio in the risk tolerance function. The other terms contain effects of correlation $\rho$, volatility of volatility  $\beta$, higher Taylor expansion terms of the stochastic Sharpe ratio, and higher derivatives of the the risk tolerance function with respect to wealth. Even for a utility function where there is no explicit solution for the constant parameter Merton value function $M$, the risk tolerance is easily computed by numerically  solving Black's equation, as detailed in \cite[Section 6.2]{fouque-sircar-zariphopoulou-2012}.   
\end{itemize}

\section{Higher Order Terms\label{higher}}
Having obtained PDEs for $\Vn{0}$ and $\Vn{1}$, we examine the higher order terms.  
An exercise in accounting shows that for all $n\geq 1$ the function $\Vn{n}(t,x,y,w)$ satisfies a linear PDE of the form
\begin{equation}
\(\frac{\pa}{\pa t} + \Ac_0+ \Bc_0\) \Vn{n} + H_n =0, \qquad \Vn{n}(T,x,y,w)=0,\label{Vneqn}
\end{equation}
where the source term $H_n$ depends only on $\Vn{k}$ $(k \leq n-1)$.  To see this, observe that the $n$th-order PDE involves three types of terms
\begin{align}
\Oc(\eps^n):&&
&\Vn{n}_t , &
&\Ac_k \Vn{n-k}, \quad ( k \leq n ) , &
\sum_{j+k+l+m=n} \chi_j \, \Vn{k}_{\alpha w}\Vn{l}_{\gamma w}\(\frac{1}{V^\eps_{ww}}\)_m, \label{eq:terms}
\end{align}
where, in the last term, $\chi$ is a place holder for one of the coefficient functions appearing in $\Nc^\eps$, the symbols $(\alpha,\del)$ are place holders for $(x,y)$ or null (meaning just a single derivative in $w$), and $\(\frac{1}{V^\eps_{ww}}\)_m$ is the $m$th order term in the Taylor series expansion of $\(\frac{1}{V^\eps_{ww}}\)$ about the point $\eps = 0$, i.e.,
\begin{align}
\(\frac{1}{V^\eps_{ww}}\)
	&=	\frac{1}{\Vn{0}_{ww}} 
			+ \sum_{k=1}^\infty \eps^k  \(\frac{1}{V^\eps_{ww}}\)_k , &
\(\frac{1}{V^\eps_{ww}}\)_k
	&=	\sum_{m=1}^k \frac{(-1)^m }{(\Vn{0}_{ww})^{1+m} }
			\( \sum_{i \in I_{k,m}} \prod_{j=1}^m \Vn{i_j}_{ww} \) , \label{eq:Taylor}
\end{align}
where $I_{k,m}$ is given by
\begin{align}
I_{k,m}
	&=	\{ i=(i_1, i_2, \cdots, i_m ) \in \mathbb{N}^m: \sum_{j=1}^m i_j = k \} . \label{eq:Ikm}
\end{align}
The terms in \eqref{Vneqn} that involve $\Vn{n}$ are precisely those terms that appear in $(\frac{\pa}{\pa t} + \Ac_0 + \Bc_0)\Vn{n}$.  The terms that do not involve $\Vn{n}$ are grouped into the source term $H_n$.  We provide here an explicit expression for the second order source term $H_2$, which appears in the $\Oc(\eps^2)$ PDE:
\begin{align}
H_2
	&=	- (\tfrac{1}{2}\lam^2)_2 \frac{(\Vn{0}_w)^2}{\Vn{0}_{ww}} 
			- ( H_1 + \Bc_0  \Vn{1})  \cdot \frac{\Vn{1}_{ww}}{\Vn{0}_{ww}} 
			+	\Ac_1 \Vn{1}
			- (\tfrac{1}{2}\lam^2)_0 \frac{(\Vn{1}_{w})^2}{\Vn{0}_{ww}}  
			\\ & \qquad
			- 2 (\tfrac{1}{2}\lam^2)_1 \frac{(\Vn{0}_w)(\Vn{1}_{w})}{\Vn{0}_{ww}}   
			- (\rho \beta \lam)_0 \frac{(\Vn{1}_{w})(\Vn{1}_{yw})}{\Vn{0}_{ww}}  
			- (\rho \beta \lam)_1 \frac{(\Vn{0}_w)(\Vn{1}_{yw})}{\Vn{0}_{ww}} 
			\\ & \qquad
			- (\tfrac{1}{2} \rho^2 \beta^2)_0 \frac{(\Vn{1}_{yw})^2}{\Vn{0}_{ww}} 
			- \mu_0 \frac{(\Vn{1}_{w})(\Vn{1}_{xw})}{\Vn{0}_{ww}}  
			- \mu_1 \frac{(\Vn{0}_w)(\Vn{1}_{xw})}{\Vn{0}_{ww}} 
			\\ & \qquad
			- (\rho \sig \beta)_0 \frac{(\Vn{1}_{xw})(\Vn{1}_{yw})}{\Vn{0}_{ww}} 
			- (\tfrac{1}{2} \sig^2)_0 \frac{(\Vn{1}_{xw})^2}{\Vn{0}_{ww}}  .
			\label{eq:H2}
\end{align}
Higher-order sources terms can be obtained systematically using a computer algebra program such as Wolfram Mathematica. 

Now let $\qn{n}$ be defined from $\Vn{n}$ by 
\begin{align}
\Vn{n}(t,x,y,w)
	&=	\qn{n}(t,x,y,z(t,w)) , \label{eq:qn}
\end{align}
using the transformation \eqref{eq:z}. Then, using Lemma \ref{princtransf}, we see that the PDE \eqref{Vneqn} for $\Vn{n}$, which has $(t,w)$-dependent coefficients through the dependence of $\Bc_0$ in \eqref{eq:B0} on $R(t,w;\lam_0)$, is transformed to the {\em constant coefficient} equation for $\qn{n}$:
\begin{equation}
\(\frac{\pa}{\pa t} + \Ac_0+ \Cc_0\)\qn{n}+Q_n=0, \qquad \qn{n}(T,x,y,z)=0. \label{qneqn}
\end{equation}
The source term is found from $H_n(t,x,y,w)=Q_n(t,x,y,z(t,w))$.

We must establish that, for every $n \geq 1$ there exists a function $Q_n$ such that $Q_n(t,x,y,z(t,w))=H_n(t,x,y,w)$.  
From \eqref{eq:terms} we see that the source term $H_n$ contains two types of terms, the first of which is
$\Ac_k \Vn{n-k}$ $(1 \leq k \leq n)$.  Since $\Ac_k$ acts only on $(x,y)$, we have that $\Ac_k \Vn{n-k} = \Ac_k \qn{n-k}$.  The second sort of term appearing in \eqref{eq:terms} are those of the form
\begin{align}
&\sum_{j+k+l+m=n}\chi_j \, \frac{\Vn{k}_{\alpha w}\Vn{l}_{\gamma w}}{\Vn{0}_{ww}}
			 \sum_{p=1}^m (-1)^p  \( \sum_{i \in I_{k,p}} \prod_{j=1}^p \frac{\Vn{i_j}_{ww}}{\Vn{0}_{ww}} \) , & 
& k,l,m \leq n-1, \label{eq:form}
\end{align}
where we have used \eqref{eq:Taylor}.

Next, using
\begin{align}
\frac{\Vn{k}_{\alpha w}\Vn{l}_{\gamma w}}{\Vn{0}_{ww}}
&= -\frac{\qn{k}_{\alpha z}\qn{l}_{\gamma z}}{\qn{0}_{z}} , &
\frac{ \Vn{i}_{ww} }{\Vn{0}_{ww}} & = \frac{-\qn{i}_{zz}}{\qn{0}_{z}} 
			+ \frac{(\qn{0}_{z}+ \qn{0}_{zz})\qn{i}_{z}}{(\qn{0}_{z})^2} ,
\end{align}
we see that \eqref{eq:form} can be written as
\begin{align}
\sum_{j+k+l+m=n}\chi_j \,\frac{-\qn{k}_{\alpha z}\qn{l}_{\gamma z}}{\qn{0}_{z}}
			 \sum_{p=1}^m (-1)^p  \( \sum_{i \in I_{k,p}} \prod_{j=1}^p \( \frac{-\qn{i_j}_{zz}}{\qn{0}_{z}} 
			+ \frac{(\qn{0}_{z}+ \qn{0}_{zz})\qn{i_j}_{z}}{(\qn{0}_{z})^2} \) \) , 
 \label{eq:q.form}
\end{align}
where $(l,k,m \leq n-1)$.
We have therefore established that, for every $n \geq 1$, the source term $H_n(t,x,y,w)$, which is composed of products and quotients of derivatives of $\Vn{k}(t,x,y,w)$ $(k \leq n-1)$, can be written be written as a function $Q_n$, which is composed of products and quotients of derivatives of $\qn{k}(t,x,y,z)$ $(k \leq n-1)$.

In Proposition \ref{q1prop}, we saw that $\qn{1}$, the first-order transformed value function, can  be expressed as a differential operator acting on $\qn{0}$, specifically $\qn{1} = \Lc_1\qn{0}$, where
\begin{equation}
 \Lc_1=\left[(T-t)(\tfrac{1}{2}\lam^2)_1(x,y)I 
 +\half(T-t)^2\left((\half\lambda^2)_{1,0}\Lc_X +(\half\lambda^2)_{0,1}\Lc_Y\right)\right]
\frac{\pa}{\pa z}.\label{opL1}
\end{equation}
We will show that, for certain utility functions $U$, each of the higher order terms $\qn{n}$ $(n \geq 2)$ can also be written as a differential operator acting on $\qn{0}$.
From Proposition \ref{qformulaprop}, we know that if the source $Q_n$ in the $n$th order PDE \eqref{qneqn} is of the form 
\eqref{Qspec}, then this will be the case. 
From \eqref{Q1formula}, we see that
\begin{equation}
Q_1=\Qc_1 q_0 \qquad \mbox{where}\qquad \Qc_1=(\tfrac{1}{2}\lam^2)_1(x,y)\frac{\pa}{\pa z}, \label{opQ1}
\end{equation}
Unfortunately, this is not always the case.  

To see this, we examine $Q_2$, the source term in the PDE for $\qn{2}$, which one can compute:
\begin{align}
Q_2
	&=	(\tfrac{1}{2}\lam^2)_2 \qn{0}_z
			- \(\frac{(\qn{0}_z + \qn{0}_{zz})\qn{1}_{z}}{(\qn{0}_z)^2} - \frac{\qn{1}_{zz}}{\qn{0}_z} \)
			( Q_1 + \Cc_0  q_1 )
						+	\Ac_1 q_1
			\\ & \qquad
			+ (\tfrac{1}{2}\lam^2)_0 \frac{(\qn{1}_{z})^2}{ \qn{0}_z}
			+ 2 (\tfrac{1}{2}\lam^2)_1 \qn{1}_{z}  
			+ (\rho \beta \lam)_0 \frac{\qn{1}_{z} \qn{1}_{yz}}{ \qn{0}_z}
			\\ & \qquad
			+ (\rho \beta \lam)_1\qn{1}_{yz}
			+ (\tfrac{1}{2} \rho^2 \beta^2)_0 \frac{\qn{1}_{yz} \qn{1}_{yz}}{ \qn{0}_z} 
			+ \mu_0 \frac{\qn{1}_{z}\qn{1}_{xz}}{ \qn{0}_z}
			\\ & \qquad
			+ \mu_1 \qn{1}_{xz}
			+ (\rho \sig \beta)_0 \frac{\qn{1}_{xz} \qn{1}_{yz}}{ \qn{0}_z}
			+ (\tfrac{1}{2} \sig^2)_0 \frac{\qn{1}_{xz} \qn{1}_{xz}}{ \qn{0}_z} . \label{eq:Q2}
\end{align}
From \eqref{eq:Q2} we see that $Q_2$ can be written as $Q_2=\Qc_2 q_0$ where
\begin{align}
\Qc_2
	&=	(\tfrac{1}{2}\lam^2)_2 \frac{\pa}{\pa z}
			- \(\frac{(\qn{0}_z + \qn{0}_{zz})\qn{1}_{z}}{(\qn{0}_z)^2} - \frac{\qn{1}_{zz}}{\qn{0}_z} \)
			( \Qc_1 + \Cc_0  \Lc_1 )
			+	\Ac_1 \Lc_1
			\\ & \qquad
			+ (\tfrac{1}{2}\lam^2)_0 \frac{(\qn{1}_{z})}{ \qn{0}_z} \frac{\pa}{\pa z} \Lc_1
			+ 2 (\tfrac{1}{2}\lam^2)_1 \frac{\pa}{\pa z} \Lc_1  
			+ (\rho \beta \lam)_0 \frac{(\qn{1}_{yz})}{ \qn{0}_z} \frac{\pa}{\pa z} \Lc_1
			\\ & \qquad
			+ (\rho \beta \lam)_1 \frac{\pa^2}{\pa y\pa z} \Lc_1
			+ (\tfrac{1}{2} \rho^2 \beta^2)_0 \frac{(\qn{1}_{yz})}{ \qn{0}_z} 
				\frac{\pa^2}{\pa y\pa z} \Lc_1
			+ \mu_0 \frac{(\qn{1}_{xz})}{ \qn{0}_z} \d_z \Lc_1
			\\ & \qquad
			+ \mu_1 \frac{\pa^2}{\pa x\pa z} \Lc_1
			+ (\rho \sig \beta)_0 \frac{(\qn{1}_{xz})}{ \qn{0}_z} \frac{\pa^2}{\pa y\pa z} \Lc_1
			+ (\tfrac{1}{2} \sig^2)_0 \frac{(\qn{1}_{xz})}{ \qn{0}_z} \frac{\pa^2}{\pa x\pa z} \Lc_1 ,	\label{eq:Qc.2}
\end{align}
where $\Qc_1$ was given in \eqref{opQ1}, and $\Lc_1$ in \eqref{opL1}.

In order to use Proposition \ref{qformulaprop}, we must establish that coefficients of $\Qc_2$ are polynomials in $(x,y,z)$. The complicating terms are those that contain derivatives of $\qn{0}$ and $\qn{1}$ divided by $\qn{0}_z$.  Such terms are always polynomials in $(x,y)$, but may not be polynomial in $z$.  The following lemma provides conditions under which
the differential operator $\Qc_n$ is guaranteed to have coefficients that are 
independent of $z$:
\begin{lemma}
\label{thm:exp}
Suppose $\qn{0}(t,z)$ is of the form:
\begin{align}
\qn{0}(t,z) = a(t)\ee^{b(t) + z c(t)}.  \label{eq:exp.form}
\end{align}
Then, for every $n \geq 1$, the 
source term $Q_n$ appearing in PDE \eqref{qneqn} can be written as $Q_n = \Qc_n \qn{0}$, where the differential operator $\Qc_n$ has coefficients that are polynomial in $(x,y)$ and independent of $z$.  
\end{lemma}

\begin{proof}
We will prove by induction on $n$ that there exists a differential operator $\Qc_n$ whose coefficients are polynomial in $(x,y)$, independent of $z$, and which satisfies $Q_n = \Qc_n \qn{0}$, where $Q_n$ is the $n$th-order source term appearing in \eqref{qneqn}.  We know from \eqref{opQ1} that such a $\Qc_1$ exists.  We now assume such $\Qc_k$ exist for all $1 \leq k \leq n-1$, and we show that $\Qc_n$ exists and has the required form.  

The existence $\Qc_k$ implies from Proposition \ref{qformulaprop} that there exists an operator $\Lc_k$ such that $\qn{k}=\Lc_k \qn{0}$.  Moreover, since $\Qc_k$ is polynomial in $(x,y)$ and independent of $z$ it follows from \eqref{qformula} that $\Lc_k$ has coefficients that are polynomial in $(x,y)$ and independent of $z$.  Now, we recall 
that $Q_n$ contains two types of terms: $\Ac_k \qn{n-k}$ $(1\leq k \leq n)$ and terms of the form \eqref{eq:q.form}.  Let us first examine terms of the form $\Ac_k \qn{n-k}$ $(1\leq k \leq n)$.  Note that
\begin{align}
\Ac_k \qn{n-k}
	&=	\Ac_k \Lc_{n-k} \qn{0} , &
1 
	&\leq k \leq n .
\end{align}
The coefficients of $\Ac_k$ are polynomial in $(x,y)$ and independent of $z$ by construction.  Hence, the coefficients of $\Ac_k \Lc_{n-k}$ are also polynomial in $(x,y)$ and independent of $z$.  Now, let us examine the terms of the form \eqref{eq:q.form}.  Using $\qn{k} = \Lc_k \qn{0}$ we can express \eqref{eq:q.form} as
\begin{align}
\sum_{j+k+l+m=n}\chi_j \, \frac{-(\frac{\pa^2}{\pa\alpha\pa z} \Lc_k \qn{0}) \, 
(\frac{\pa^2}{\pa\gamma\pa z}\Lc_l \qn{0})}{\qn{0}_z}
			 \sum_{p=1}^m (-1)^p  \( \sum_{i \in I_{k,p}} \prod_{j=1}^p \( \frac{(\qn{0}_z + \qn{0}_{zz})(\d_z \Lc_{i_j}\qn{0})}{(\qn{0}_z)^2} - \frac{\frac{\pa^2}{\pa z^2} \Lc_{i_j} \qn{0}}{\qn{0}_z} \) \) ,  \label{eq:temp}
\end{align}
where $l,k,m \leq n-1$.  Since, by assumption, $\qn{0}$ is of the form \eqref{eq:exp.form}, it follows that
terms of the form \eqref{eq:temp} are polynomials in $(x,y)$ and independent of $z$.
We have therefore established that $Q_n$ can be written as $\Qc_n \qn{0}$ where $\Qc_n$ is a differential operator whose coefficients are polynomial in $(x,y)$ and independent of $z$.
\end{proof}

We will see in the next section that, when $U$ belongs to the power utility class, then $\qn{0}$ is of the form \eqref{eq:exp.form}.  Thus, the $n$th-order term $\qn{n}$ can be written as a differential operator $\Lc_n$ acting on $\qn{0}$. 

\section{Specific results for power utility}
\label{sec:examples}
In this section, we consider the case where the utility function $U$ belongs to the power utility class
\begin{align}
\text{Power utility:}&&
U(w)
	&=	\frac{w^{1-\gam}}{1-\gam} , &
\gam
	&>0, \,\,\gam\neq1 . \label{eq:CRRA}
\end{align}
For general LSV dynamics \eqref{dXdY}, we will obtain the second-order approximation for the value function $u$, optimal investment strategy $\pi^*$ and implied Sharpe ratio $\Lam$.  Then, we will establish error estimates for the approximate value function in a stochastic volatility setting.

\subsection{Value function}
To obtain the second order approximation to the value function $V$, we must first compute $\qn{0}$, $\qn{1}$ and $\qn{2}$.  
With $U$ given by \eqref{eq:CRRA}, we have $U'(w) = w^{-\gam}$ and $[U']^{-1}(\zeta)=\zeta^{-1/\gam}$.  
As found in \cite{merton1969lifetime}, we have 
\begin{align}
\Vn{0}(t,w)
	&=M(t,w;\lam_0)=	\frac{w^{1-\gam}}{1-\gam}\,\exp\(\tfrac{1}{2}\(\frac{1-\gam}{\gam}\)\lam^2_0(T-t)\)  .\label{Mertonftn}
\end{align} 
Then the transform variable $z$ in \eqref{eq:z} is given by 
\begin{align}
z(t,w)
	&=	\gam \log w + \tfrac{1}{2}\(\tfrac{2 \gam - 1}{\gam}\)\lam^2_0  (T-t), \label{eq:z.CRRA}
\end{align}
and the solution of the heat equation PDE problem \eqref{q0PDE1} is 
\begin{align}
\qn{0}(t,z)
		&=		\frac{1}{1-\gam}
			\exp \( \frac{1-\gam}{\gam} \( z + \tfrac{1}{2}\(\tfrac{1-\gam}{\gam}\)\lam^2_0(T-t)\) \) . \label{eq:q0.crra}
\end{align}
Next, using \eqref{q1formula} and \eqref{eq:q0.crra}, we compute $\qn{1}$: 
\begin{align}
\qn{1}(t,x,y,z)
	&=	\frac{1-\gam}{\gam} 
			\bigg(
				\( \tfrac{1}{2} \lam^2 \)_{1,0} \( (T-t)(x-\xb) + \tfrac{1}{2}(T-t)^2 
				\( \tfrac{1}{\gam} \mu_0 - \tfrac{1}{2} \sig_0^2 \) \) \\ 
				&\qquad \qquad
				+\( \tfrac{1}{2} \lam^2\)_{0,1} \( (T-t)(y-\yb) + \tfrac{1}{2}(T-t)^2 
				\( c_0 + \tfrac{1-\gam}{\gam} \rho \beta_0 \lam_0 \) \) 
			\bigg)\qn{0}(t,z) ,	\label{eq:q1.CRRAa}
\end{align}
from which we obtain 
\begin{align}
\Vn{1}(t,x,y,z)
	&=	\frac{1-\gam}{\gam} 
			\bigg(
				\( \tfrac{1}{2} \lam^2 \)_{1,0} \( (T-t)(x-\xb) + \tfrac{1}{2}(T-t)^2 
				\( \tfrac{1}{\gam} \mu_0 - \tfrac{1}{2} \sig_0^2 \) \)\label{eq:u1.CRRA} \\ 
				&\qquad \qquad
				+\( \tfrac{1}{2} \lam^2\)_{0,1} \( (T-t)(y-\yb) + \tfrac{1}{2}(T-t)^2 
				\( c_0 + \tfrac{1-\gam}{\gam} \rho \beta_0 \lam_0 \) \) 
			\bigg)\Vn{0}(t,z) ,	
\end{align}
where, as a reminder, $(\tfrac{1}{2}\lam^2)_{1,0}$ and $(\tfrac{1}{2}\lam^2)_{0,1}$ are given by
\begin{align}
(\tfrac{1}{2}\lam^2)_{1,0}
	&= 	 \( \frac{\mu^2}{2\sig^2} \)_x(\xb,\yb) , &
(\tfrac{1}{2}\lam^2)_{0,1}
	&=	\( \frac{\mu^2}{2\sig^2} \)_y(\xb,\yb) .
\end{align}

Having obtained an explicit expressions for $\qn{0}$ and $\qn{1}$, we can now compute $\qn{2}$. The second order source term $Q_2$, given by \eqref{eq:Q2}, can be written as $Q_2 = \Qc_2 q_0$ where the operator $\Qc_2$ is given by \eqref{eq:Qc.2}.  From \eqref{eq:q0.crra}, we see that $\qn{0}$ is of the form \eqref{eq:exp.form}.  Thus, from Lemma \ref{thm:exp} we know that the coefficients of $\Qc_2$ are polynomial in $(x,y)$ and independent of $z$.  Therefore, we can use Proposition \ref{qformulaprop} to compute $\qn{2} = \Lc_2 q_0$. The expression for $\qn{2}$ is quite long.  As such, for the sake of brevity, we do not include it here. 

We can obtain $\Vn{2}$ from $\qn{2}$ using \eqref{eq:q} and \eqref{eq:z.CRRA}. The same procedure can be used to compute higher-order terms: $\qn{n}$ $(n \geq 3)$.  Since the expressions for $\un{2}$ and higher-order terms are quite long, we do not present them here.  However, in the numerical examples that follow, we do compute the second order approximation, and we will see that it provides a noticeably more accurate approximation of $V$ than does the first order approximation.

\subsection{Optimal Strategy}
For power utility, the Merton risk tolerance function is especially simple: $R(t,w;\lam)=w/\gam$, and it does not depend on $t$, $T$ or the Sharpe ratio $\lam$. Therefore, the approximate first order optimal strategy in \eqref{pistarapprox1} is given by $ \pi^*\approx \bar\pi$, where
\begin{equation}
\bar\pi(t,x,y,w) = \left[\frac{\mu(x,y)}{\sig^2(x,y)} + 
(T-t)\lam_0\(\frac1\gam-1\)\left(\frac{\rho\beta(x,y)}{\sig(x,y)}\lam_{0,1}+\lam_{1,0}\right)\right]\,\frac{w}{\gam}, \label{pistarapproxpower}
\end{equation}
which is also proportional to the current wealth level $w$ as in the classical Merton strategy, but with proportion that varies with the model coefficients whose values move with the log stock price $x$ and the volatility driving factor $y$.  

We remark that it is also possible to compute the next order of the strategy approximation $\pi_2^*$ in the case of power utility using the lengthy expression for $\Vn{2}$.  
For the special case $(x,y)=(\xb,\yb)$, we have
\begin{align}
\pi_2^*
	&=	w\times\left[ \frac{(T-t)^2(\gam-1)}{2 \gam^3}
			\Big(
			(\gam - 1 ) (\tfrac{1}{2}\lam^2)_{0,1} (\rho \beta \lam)_{1,0}+\gamma  (\tfrac{1}{2}\lam^2)_{1,0} (\tfrac{1}{2} \sig^2)_{1,0}\right.
			\\ 
			& \qquad
			-\left(\gamma  c_{1,0} (\tfrac{1}{2}\lam^2)_{0,1}+\left(\gamma  c_0-(\gam - 1 ) (\rho \beta \lam)_0\right) (\tfrac{1}{2}\lam^2)_{1,1}+2 \mu _0 (\tfrac{1}{2}\lam^2)_{2,0}-\gamma  \sigma _0^2 (\tfrac{1}{2}\lam^2)_{2,0}+(\tfrac{1}{2}\lam^2)_{1,0} \mu _{1,0}\right)
			\Big)
			\\ & \qquad  \left.
			+ \frac{(T-t)^3(\gam-1)^2}{8 \gam^4 (\tfrac{1}{2}\sig^2)_0}
			\Big(
			(\rho \sig \beta)_0 (\tfrac{1}{2}\lam^2)_{0,1} \left(-2 \left(\gamma  c_0-(\gam - 1 ) (\rho \beta \lam)_0\right) (\tfrac{1}{2}\lam^2)_{0,1}+\left(-2 \mu _0+\gamma  \sigma _0^2\right) (\tfrac{1}{2}\lam^2)_{1,0}\right)
			\Big)\right] . \\[-4em]
\end{align}

\subsection{Implied Sharpe ratio}
We now compute the first order approximation  of the implied Sharpe ratio $\Lam$, which was introduced in Section \ref{ISR}.  
From \eqref{ISRsumm}, we have $\Lam \approx \bar\Lam$, where
\begin{equation}
 \bar\Lam= \lam_0 + \lambda_{1,0}(x-\bar x) + \lambda_{0,1}(y-\bar y) +\frac12(T-t)
			\Big(\lam_{0,1} \(c_0+ \tfrac{1-\gam}{\gam} (\rho \beta \lam)_0\)
						+ \lam_{1,0} \( \tfrac{1}{\gam}\mu_0-(\tfrac{1}{2}\sig_0^2) \) \Big). \label{eq:Lam.CRRA}
\end{equation}
The second order correction 
$\Lam_2$ is quite long, and we omit it for the sake of brevity.

Observe that, for power utility, in which an explicit expression for the constant parameter Merton value function $M$ is available, one can obtain an expression for the implied Sharpe ratio $\Lam$ by solving \eqref{eq:implied.sharpe} with $M$ given by \eqref{Mertonftn}:
\begin{align}
\Lam
	&=	\sqrt{ \log \( \frac{V}{U(w)} \) \frac{2 \gam}{(1-\gam)(T-t)} } . \label{eq:sharpe.exact}
\end{align}
This will be useful when we test the numerical accuracy of the Sharpe ratio approximation in two examples.

\subsection{Accuracy of the approximation for stochastic volatility models} 
In this section, we establish the accuracy of 
$$\Vnb{n}=\sum_{k=0}^n\Vn{k},$$
the $n$th-order approximation of the value function $V$, assuming stochastic volatility dynamics of the form
\begin{align}
\dd X_t
	&=	\( \mu(Y_t) - \frac{1}{2}\sig^2(Y_t) \)  \dd t + \sig(Y_t)  \dd B_t^X , \\
\dd Y_t
	&=	c(Y_t) \dd t + \beta(Y_t) \dd B_t^Y , \label{eq:SV} \\
\dd \<B^X,B^Y\>_t
	&=	\rho \,\dd t ,
\end{align}
and a utility function $U$ of the power utility class \eqref{eq:CRRA}.

Throughout this section, we will make the following assumption:
\begin{assumption}
\label{assumption}
There exists a constant $C>0$ such that the following holds: \\
(i) \emph{Uniform ellipticity}: $1/C \leq \beta^2 \leq C$.\\
(ii) \emph{Regularity and boundedness}: The coefficients 
$c$, $\rho \beta \lam$, $\beta^2$ and $\lam^2$ are
$C^{n+1}(\Rb)$ and all derivatives up to order $n$ are bounded by $C$.\\
(iii) The \emph{risk aversion parameter} in the utility function \eqref{eq:CRRA} satisfies $\gam>1$.
\end{assumption}
Clearly, stochastic volatility dynamics \eqref{eq:SV} are a special case of the more general local-stochastic volatility dynamics \eqref{dXdY}.  As such, one can obtain a series approximation $\Vnb{n}$ of the value function $V^\eps(t,y,w)$, which is in this case independent of $x$, 
by solving the sequence of PDEs \eqref{Vneqn}.
An alternative but equivalent approach is to linearize the full PDE \eqref{eq:u.eps.pde},
and then perform a series approximation on the resulting linear PDE.  This is the approach we follow here.

Assuming power utility \eqref{eq:CRRA} and dynamics given by \eqref{eq:SV}, \cite{zariphopoulou2001} shows that the function 
$V^\eps(t,y,w)$, solution of \eqref{eq:u.eps.pde},
is given by
\begin{align}
V^\eps(t,w,y)
	&=	\frac{w^{1-\gam}}{1-\gam} \( \psi^\eps(t,y) \)^\eta , &
\eta
	&=	\frac{\gam}{\gam + (1-\gam)\rho^2} , \label{eq:eta}
\end{align}
where the function $\psi^\eps$ satisfies the Cauchy problem
\begin{align}
0
	&=	( \d_t + \Acb^\eps ) \psi^\eps , &
\psi^\eps(T,y)
	&=	1 , &
\eps
	&\in [0,1] , \label{eq:psi.eps.pde}
\end{align}
and $\Acb^\eps$ is a linear elliptic operator given by
\begin{align}
\Acb^\eps
	&=	(\tfrac{1}{2} \beta^2)^\eps \d_y^2 + \( c^\eps + \tfrac{1-\gam}{\gam} (\rho \beta \lam)^\eps  \) \d_y 
			+ \tfrac{1-\gam}{\eta \gam}(\tfrac{1}{2}\lam^2)^\eps . \label{eq:A.hat}
\end{align}
Let us denote $\Acb = \Acb^\eps |_{\eps=1}$ and $\psi = \psi^\eps |_{\eps=1}$.

\begin{remark}
Assumption \ref{assumption} part (iii) guarantees that the last term in \eqref{eq:A.hat} is strictly negative.
\end{remark}
\begin{remark}
The linearization transformation described above works only for one-factor pure stochastic volatility dynamics \eqref{eq:SV}, or complete market pure local volatility models, and only for power utility \eqref{eq:CRRA}.  For more general local-stochastic volatility dynamics \eqref{dXdY} and utility functions $U$, one must work with nonlinear PDE \eqref{eq:hjb.3}.
\end{remark}

We return now to \eqref{eq:psi.eps.pde}.  Noting that $\Acb^\eps$ can be written as
\begin{align}
\Acb^\eps
	&=	\sum_{n=0}^\infty \eps^n \Acb_n , &
\Acb_n
	&=	(\tfrac{1}{2} \beta^2)_n \d_y^2 + \( c_n + \frac{1-\gam}{\gam} (\rho \beta \lam)_n  \) \d_y 
			+ \frac{1-\gam}{\eta \gam} (\tfrac{1}{2}\lam^2)_n , \label{eq:Acb.eps.expand}
\end{align}
we seek a solution $\psi^\eps$ to \eqref{eq:psi.eps.pde} of the form
\begin{align}
\psi^\eps
	&=	\sum_{n=0}^\infty \eps^n \psi_n . \label{eq:psi.eps.expand}
\end{align}
Inserting \eqref{eq:Acb.eps.expand} and \eqref{eq:psi.eps.expand} into PDE \eqref{eq:psi.eps.pde} and collecting terms of like powers of $\eps$ we obtain the following sequence of nested PDEs:
\begin{align}
\Oc(1):&&
0
	&=	( \d_t + \Acb_0 ) \psi_0 , &
\psi_0(T,y)
	&=	1 ,  \label{eq:psi0.pde} \\
\Oc(\eps^n):&&
0
	&=	( \d_t + \Acb_0 ) \psi_n + \sum_{k=1}^n \Acb_k \psi_{n-k}, &
\psi_n(T,y)
	&=	0 .  \label{eq:psin.pde}
\end{align}
This sequence of nested PDEs has been solved explicitly in \cite{lorig-pagliarani-pascucci-2}.  We present the result here.
\begin{theorem}
\label{thm:psi}
Let $\psi_0$ and $\psi_n$ $(n \geq 1)$ satisfy \eqref{eq:psi0.pde} and \eqref{eq:psin.pde}, respectively.  Then, omitting $y$-dependence for simplicity, we have
\begin{align}
\psi_0(t)
	&=	\exp\( (T-t) \frac{1-\gam}{\eta \gam} (\tfrac{1}{2}\lam^2)_0 \) , &
\psi_n(t)
	&=	\Lch_n(t,T) \psi_0(t) , \label{eq:psi0}
\end{align}
where the linear operator $\Lch_n(t,T)$ is given by
\begin{align}
\Lch_n(t,T)
	&:=	\sum_{k=1}^n \int_t^T \dd t_1 \int_{t_1}^T \dd t_2 \cdots \int_{t_{k-1}}^T \dd t_k \sum_{I_{n,k}}
			\Gch_{i_1}(t,t_1) \Gch_{i_2}(t,t_2) \cdots \Gch_{i_k}(t,t_k) , 
\end{align}
with $I_{n,k}$ defined in \eqref{eq:Ikm} and
\begin{align}
\Gch_{i}(t,t_k)
	&=	\Acb_i(\Ych(t,t_k)) , &
\Ych(t,t_k)
	&=	y + (t_k-t) \( c_0 + \frac{1-\gam}{\gam} (\rho \beta \lam)_0  \) + 2 (\tfrac{1}{2} \beta^2)_0 \d_y . 
\end{align}
Here, the notation $\Acb_i(\Ych(t,t_k))$ indicates that $y$ is replaced by $\Ych(t,t_k)$ in the coefficients of $\Acb_i$.
\end{theorem}
\begin{proof}
See \cite[Theorem 7]{lorig-pagliarani-pascucci-2}.
\end{proof}
\noindent
Having obtained an explicit expression for $\psi_n$ $(n \geq 0)$ we now define $\psib_n$, the $n$th-order approximation of $\psi$:
\begin{align}
\psib_n
	&:= \sum_{k=0}^n \psi_k ,\quad\mbox{with}\quad \bar y =y. \label{eq:psib.n}
\end{align}
The accuracy of the series approximation $\psib_n$ 
is established in \cite{lorig-pagliarani-pascucci-4}.
\begin{theorem}
Let $\psi$ be the solution of
\eqref{eq:psi.eps.expand} with $\eps=1$
and let $\psib_n$ be defined by \eqref{eq:psib.n} with $\psi_i$ $(i \geq 0)$ as given in Theorem \ref{thm:psi}.  Then, under Assumption \ref{assumption}, we have
\begin{align}
\sup_y |\psi(t,y) - \psib_n(t,y)|
	&=	\Oc ( \tau^{\frac{n+3}{2}} ) , &
\tau
	&:=	T-t . \label{eq:psi.accuracy}
\end{align}
\end{theorem}
\begin{proof}
See \cite[Theorem 3.10]{lorig-pagliarani-pascucci-4}.
\end{proof}
\noindent
Our task is now to translate the approximation $\psib_n$ and accuracy result for $|\psi-\psib_n|$ into an approximation $\Vnb{n}$ and accuracy result for $|V-\Vnb{n}|$.  
Expanding $V^\eps$, given by \eqref{eq:eta}, in powers of $\eps$, we obtain
\begin{align}
V^\eps
	&=	\frac{w^{1-\gam}}{1-\gam} \( \psi^\eps \)^\eta \\
	&=	\frac{w^{1-\gam}}{1-\gam}  \psi_0^\eta 
			+ \frac{w^{1-\gam}}{1-\gam} \sum_{k=1}^\infty \eps^k \( \sum_{m=1}^k \frac{1}{m!} \( \d_\psi^m \psi_0^\eta \) 
			\( \sum_{i \in I_{k,m}} \prod_{j=1}^m \psi_{i_j} \) \) 
	=: \sum_{k=0}^\infty \eps^k V_k ,
\end{align}
where $I_{k,m}$ is defined in \eqref{eq:Ikm}, and 
\begin{align}
\Vnb{n}
	&=	\sum_{k=0}^n \Vn{k} , &
\Vn{0}
	&=	\frac{w^{1-\gam}}{1-\gam}  \psi_0^\eta , &
\Vn{k}
	&=	\frac{w^{1-\gam}}{1-\gam} \sum_{m=1}^k \frac{1}{m!} \( \d_\psi^m \psi_0^\eta \) 
			\( \sum_{i \in I_{k,m}} \prod_{j=1}^m \psi_{i_j} \) , \label{eq:u0.un} 
		\end{align}
		and $\bar y$ is set by $\bar y =y$.
The following theorem establishes the accuracy of $\Vnb{n}$, the $n$th-order approximation of the value function $V$.
\begin{theorem}
\label{thm:u}
Let $(X,Y)$ have stochastic volatility dynamics \eqref{eq:SV} and assume the utility function $U$ is of the  power utility class \eqref{eq:CRRA}.  Then, under Assumption \ref{assumption}, for a fixed $w$, the approximate value function $\Vnb{n}$, given by \eqref{eq:u0.un}, satisfies
\begin{align}
\sup_y |V(t,y,w) - \Vnb{n}(t,y,w) |
	&=	\Oc(\tau^{\frac{n+3}{2}}) , & 
\tau
	&:=	T-t , \label{eq:u.accuracy}
\end{align}
where $\eta$ is defined in \eqref{eq:eta}.
\end{theorem}
\begin{proof}
Theorem \ref{thm:psi} implies that $\psi_0(t) = \Oc(1)$ as $\tau \to 0$ and equation \eqref{eq:psi.accuracy} implies
\begin{align}
\sup_y \psi_n(t,y)
	&=	\Oc( \tau^{\frac{n+2}{2}} ) , &
n
	&\geq 1 . \label{eq:orders}
\end{align}
It therefore follows from \eqref{eq:u0.un} that
$V_n$ 
satisfies
\begin{align}
\sup_y V_n(t,y,w)
	&=	\Oc( \tau^{\frac{n+2}{2}} ) , &
n
	&\geq 1 .
\end{align}
Therefore, we have
\begin{align}
\sup_y |V(t,y,w) - \Vnb{n}(t,y,w) |
	&=	\Oc(\tau^{\frac{n+3}{2}}) ,
\end{align}
as claimed.
\end{proof}

\section{Examples}
\label{eq:examples2}
In this section we provide two numerical examples, which illustrate the accuracy and versatility of the series approximations developed in this paper. Both are based on power utility, but the first order approximations described in Section \ref{summ} could be computed for utility functions outside of this class, for instance mixture of power utilities, introduced in \cite{fouque-sircar-zariphopoulou-2012}, which allow for wealth-varying relative risk aversion. There the solution of the constant parameter Merton problem $M$ is computed numerically, and LSV corrections in the formulas of Section \ref{summ} can be obtained by numerical differentiation. 

\subsection{Stochastic volatility example}
In our first example, we consider a stochastic volatility model in which the coefficients $(\mu,\sig,c,\beta)$ appearing in \eqref{dXdY} are given by
\begin{align}
\mu(y)
	&=	\mu , &
\sig(y)
	&=	\frac{1}{\sqrt{y}} , &
c(y)
	&=	\kappa ( \theta - y ) , &
\beta(y)
	&=	\del \sqrt{y} . \label{eq:stoch.vol}
\end{align}
Here, the constants $(\kappa, \theta, \del)$ must satisfy the usual Feller condition: 
$2 \kappa \theta \geq \del^2$. 

Assuming power utility \eqref{eq:CRRA}, 
an explicit formula for the value function of the infinite horizon consumption problem is obtained in \cite{chacko-viceira-2005}.  For the terminal utility optimization problem that we consider in this paper, an explicit formula for the value function $V$ 
in \eqref{eq:v.def}
is obtained in \cite[Section 6.4]{fouque-sircar-zariphopoulou-2012}:
\begin{align}
V(t,y,w)
	&=	\( \frac{w^{1-\gam}}{1-\gam} \) \ee^{\eta A(T-t) y + \eta B(T-t)} , &
\eta
	&=	\frac{\gamma }{\gamma +(1-\gamma )\rho^2} , \label{eq:v.heston} \\
A(t)
	&= a_+ \frac{1 - \ee^{-\alpha t}}{1-(a_-/a_+)\ee^{-\alpha t}}, &
B(t)
	&=	\kappa \theta \( a_- t - \frac{2}{\del^2} 
			\log \( \frac{1-(a_-/a_+)\ee^{-\alpha t}}{1-(a_-/a_+)}\) \), 
\end{align}
where
\begin{align}
a_\pm
	&=	\frac{-q \pm \sqrt{q^2-4 pr}}{2p} , &
\alpha 
	&=	\sqrt{q^2- 4 p r} , \\
p
	&=	\frac{1}{2} \del^2 , &
q
	&= 	\del \(\frac{1-\gamma }{\gamma }\) \mu \rho - \kappa, &
r
	&=	\frac{1}{2}\(\frac{1-\gamma }{\eta \gamma }\)\mu^2.
\end{align}
An explicit formula for the optimal investment strategy $\pi^*$ can be obtained by inserting \eqref{eq:v.heston} into \eqref{eq:pi}.
Likewise, an explicit formula for the implied Sharpe ratio $\Lam$ can be obtained by inserting \eqref{eq:v.heston} into \eqref{eq:sharpe.exact}.

The zeroth, first and second order approximations for $V$, $\pi^*$ and $\Lam$ can be obtained using the results of Section \ref{sec:examples}.
Fixing $\yb=y$ and using \eqref{Mertonftn}, \eqref{eq:u1.CRRA}, \eqref{pistarapproxpower} and \eqref{eq:Lam.CRRA}, we obtain
\begin{align}
\Vn{0}
	&=	U(w) \exp\( \frac{1-\gam}{\gam} \frac{\mu^2 y}{2} (T-t) \), &
\Vn{1}
	&=	\frac{1-\gam}{4\gam^2} \mu^2  (T-t)^2 
			\Big(\gamma  \kappa (\theta - y) +(1-\gam) \rho \del \mu y \Big) \Vn{0} , \\
\pi^*
	&\approx	\frac{w}{\gam} \mu y +	\frac{w}{\gam} \rho \del \mu^2 \frac{1-\gam}{\gam} (T-t) y  , \\
\Lam_0
	&=	\mu \sqrt{y}, &
\Lam_1
	&=	\frac{(T-t)}{4 \gamma \sqrt{y}}
			\mu \Big(\gamma  \kappa (\theta - y) +(1-\gam) \rho \del \mu y \Big),
\end{align}
where in the strategy we have also expanded the coefficients in Taylor series. 
The second-order terms are omitted for the sake of brevity.

In Figure \ref{fig:heston} we plot as a function of $\sig=1/\sqrt{y}$ the exact value function $V$, the exact optimal investment strategy $\pi^*$ and the exact implied Sharpe ratio $\Lam$.  We also plot the zeroth, first and second-order approximations of these quantities.  For all three quantities, we observe a close match between the exact function ($u$, $\pi^*$ and $\Lam$) and the second order approximation.  Figure \ref{fig:heston} also includes a plot of the implied Sharpe ratio (both exact $\Lam$ and the second order approximation) as a function of the risk-aversion parameter $\gam$ for three different time horizons.  It is clear from the figure that the approximation is most accurate at the shortest time horizons, consistent with the accuracy result of Theorem \ref{thm:u}.
\begin{figure}[htb]
\centering
\begin{tabular}{cc}
\includegraphics[width=0.475\textwidth]{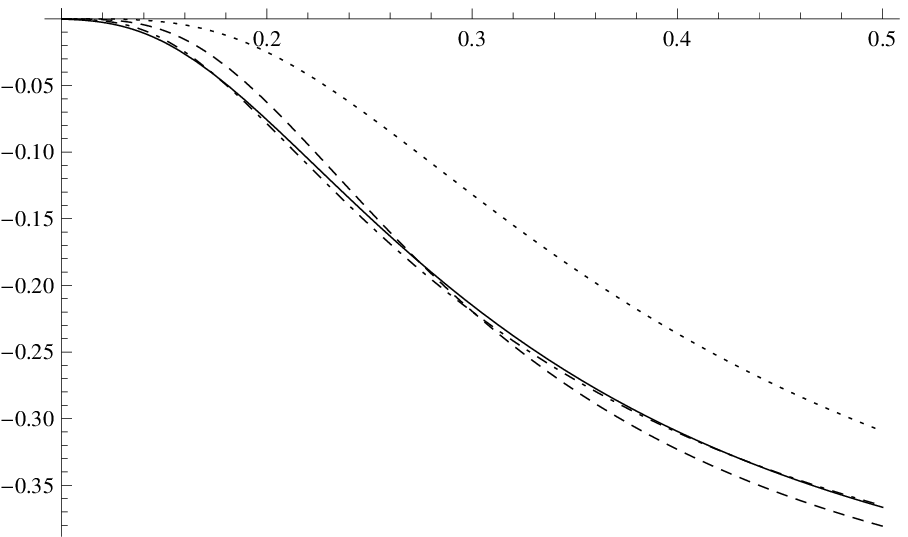}&
\includegraphics[width=0.475\textwidth]{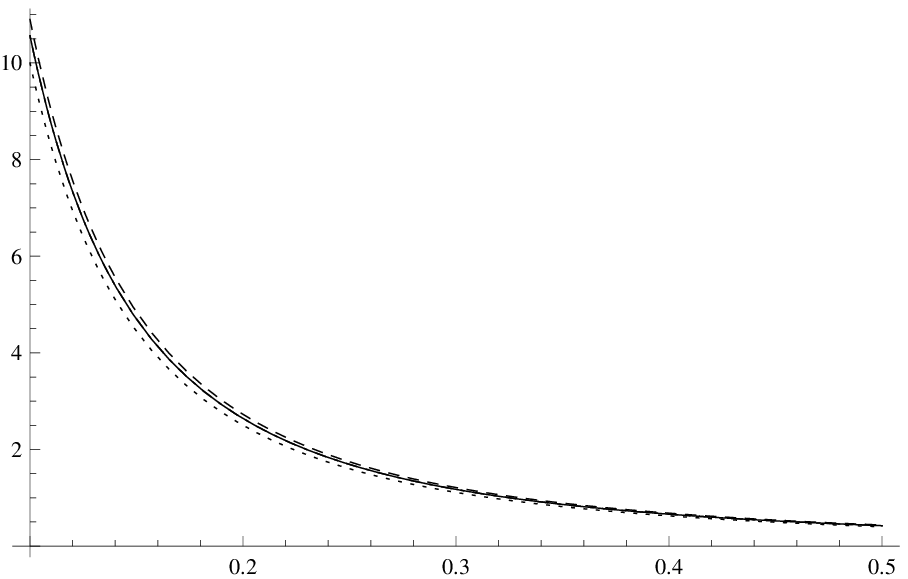}\\
$V$ vs $\sigma$ & $\pi^*$  vs $\sigma$ \\
\includegraphics[width=0.475\textwidth]{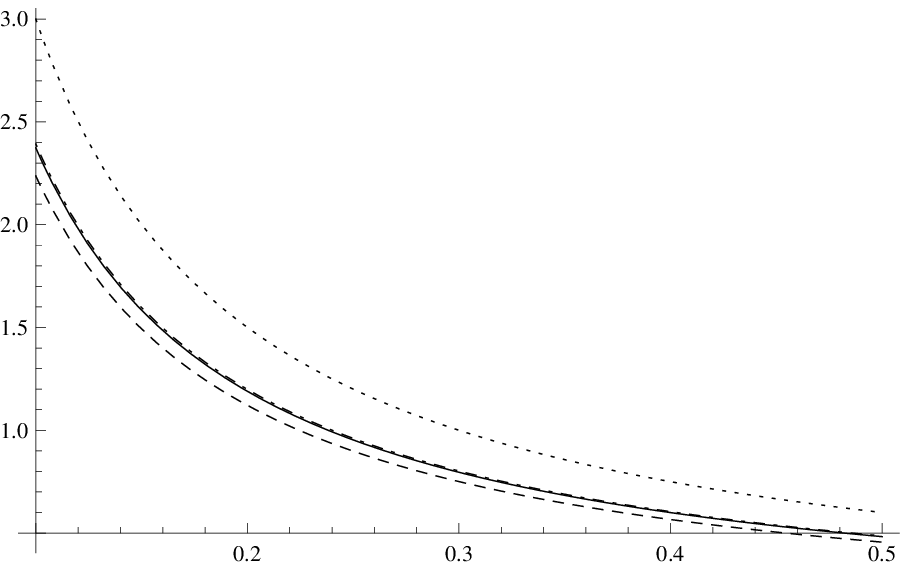}&
\includegraphics[width=0.475\textwidth]{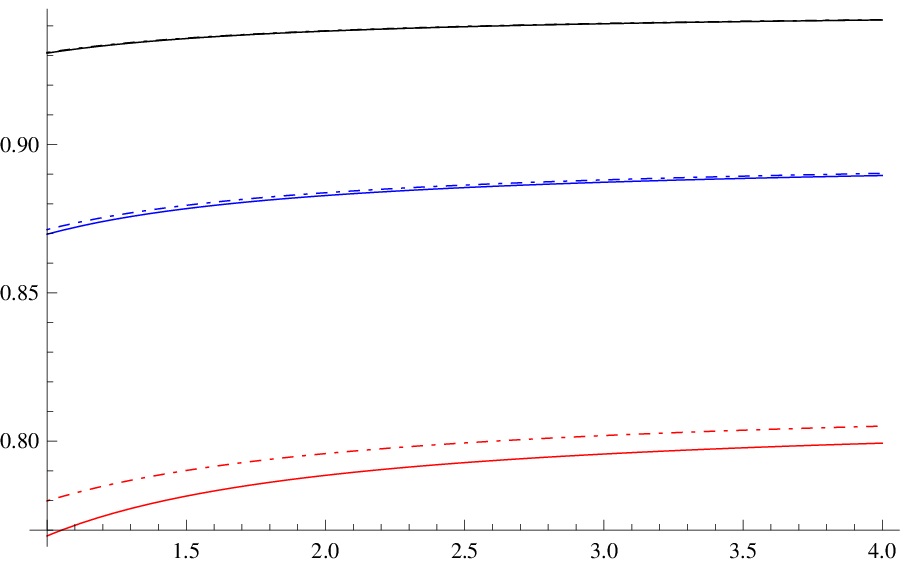}\\
$\Lam$  vs $\sigma$ & $\Lam$  vs $\gamma$
\end{tabular}
\caption{\small{{\em The value function (top left), optimal strategy  (top right) and implied Sharpe ratio (bottom left) are plotted as a function of instantaneous level of volatility $\sig=1/\sqrt{y}$ assuming power utility \eqref{eq:CRRA} and dynamics given by stochastic volatility model \eqref{eq:stoch.vol}.  
In all three plots, the solid line corresponds to the exact function and the dotted, dashed and dot-dashed lines correspond to the zeroth, first and second-order approximations, respectively.
The parameters used in these three plots are: $T-t=4.0$, $w=1.0$, $\kappa = 0.3$, $\theta = 0.2$, $\del = 0.3$, $\rho=-0.75$, $\mu = 0.3$ and $\gam = 3.0$.  On the bottom right we fix the volatility $\sig=0.3$ and we plot the implied Sharpe ration $\Lam$ as a function of $\gam$ for three different time horizons $T-t=\{1,2,4\}$ corresponding to black, blue and red, respectively.  The solid lines are exact.  The dot-dashed lines correspond to the second-order approximation.  The parameters used in the bottom right plot are $w=1.0$, $\kappa = 0.3$, $\theta = 0.2$, $\del = 0.3$, $\rho=-0.75$ and $\mu = 0.3$.}}
}
\label{fig:heston}
\end{figure}

\subsection{Local volatility example}
We now consider a local volatility model in which the coefficients $(\mu,\sig)$ appearing in \eqref{dXdY} are given by
\begin{align}
\mu(x)
	&=	\mu , &
\sig(x)
	&=	\del \ee^{\eta x} . 
	\label{eq:local.vol}
\end{align}
Since $Y$ plays no role in the dynamics of $X$, the coefficients $c$ and $\beta$ do not appear.

Assuming power utility \eqref{eq:CRRA}, an explicit formula for the value function $V$ in this setting is obtained in \cite{dries2005}
\begin{align}
V(t,x,w)
	&=	\frac{w^{1-\gam}}{1-\gam} \(f(t,\ee^{-2 \eta x})\)^\gam , &
f(t,d)
	&=	A(t) \ee^{B(t) d} , \label{eq:v.cev} \\
A(t)
	&=	\ee^{\lam_+ \eta(2\eta+1)(T-t)}
		\( \frac{\lam_- - \lam_+}{\lam_- - \lam_+ \ee^{2\eta^2(\lam_+-\lam_-)(T-t)} } \)^{\tfrac{2 \eta + 1}{2 \eta}} , &
B(t)
	&=	\frac{1}{\del^2} I(t) , \\
I(t)
	&=	\frac{\lam_+ \( 1 - \ee^{2 \eta^2(\lam_+-\lam_-)(T-t)} \)}
			{1-(\lam_+/\lam_-)\ee^{2 \eta^2(\lam_+-\lam_-)(T-t)}} , &
\lam_\pm
	&=	\frac{\mu \pm \sqrt{\gam \mu^2}}{2 \eta \gam} .
\end{align}
The optimal investment strategy $\pi^*$ can be obtained by inserting \eqref{eq:v.cev} into \eqref{eq:pi}
An explicit expression for the implied Sharpe ratio $\Lam$ can be obtained by inserting \eqref{eq:v.cev} into \eqref{eq:sharpe.exact}.

The zeroth, first and second order approximations for $V$, $\pi^*$ and $\Lam$ can be obtained using the results of Section \ref{sec:examples}.
Fixing $\xb=x$ and using \eqref{Mertonftn}, \eqref{eq:u1.CRRA}, \eqref{pistarapproxpower} and \eqref{eq:Lam.CRRA}, we obtain
\begin{align}
\Vn{0}
	&=	U(w)\exp\(\frac{1-\gam}{2 \gam} \frac{\mu^2}{\del^2}\ee^{- 2 \eta x}(T-t)\) , &
\Vn{1}
	&=	\frac{-(1-\gam)(T-t)^2\eta \mu^2}{2\gam^2\del^2} 
			\ee^{-2 \eta x} \Big( \mu -\gamma  \tfrac{1}{2} \del^2 \ee^{2 \eta x} \Big) u_0 , \\
\pi^*&\approx	w \frac{\mu}{\gam \del^2} \ee^{-2 \eta x}	- w \frac{(1-\gam)(T-t)}{\gam^2 } 
			\frac{\eta \mu^2}{ \del^2} \ee^{- 2 \eta x}  , \\
\Lam_0
	&=	\frac{\mu}{\del} \ee^{- \eta x } , &
\Lam_1
	&=	\frac{- (T-t)}{2 \gamma  }
			\frac{\eta \mu}{\del} \ee^{- \eta x} \Big( \mu -\gamma  \tfrac{1}{2} \del^2 \ee^{2 \eta x} \Big) .
\end{align}
Once again, the second-order terms are omitted for the sake of brevity.

In Figure \ref{fig:cev-2}, we plot as a function of $\sig=\del \ee^{\eta x}$ the value function $V$, the optimal investment strategy $\pi^*$ and the implied Sharpe ratio $\Lam$.  We also plot the zeroth, first and second order approximations of these quantities.  For all three quantities, we observe a close match between the exact functions ($V$, $\pi^*$ and $\Lam$) and their second order approximation.  Figure \ref{fig:cev-2} also contain plots of the implied Sharpe ratio $\Lam$ (both exact $\Lam$ and the second order approximation $\bar{\Lam}_2$) as a function of the risk-aversion parameter $\gam$ for three different time horizons.

\begin{figure}[htb]
\centering
\begin{tabular}{cc}
\includegraphics[width=0.475\textwidth]{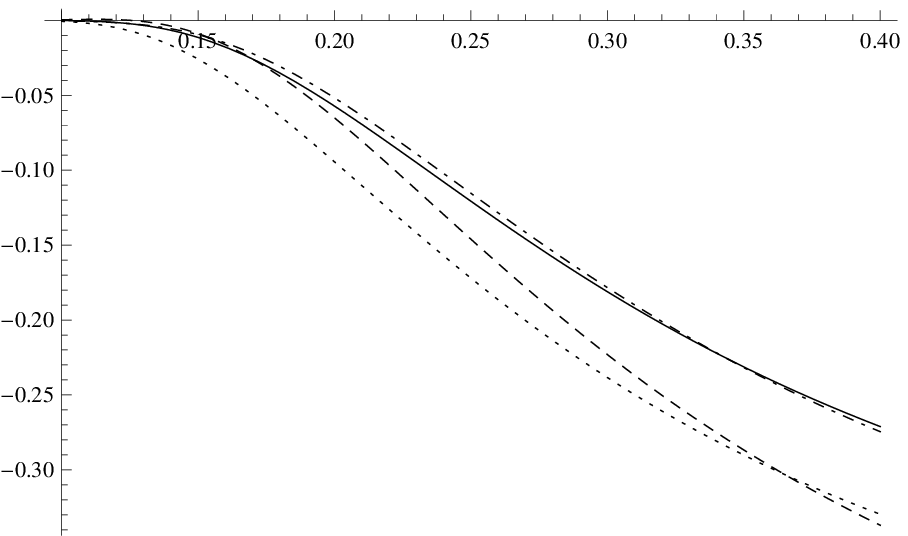}&
\includegraphics[width=0.475\textwidth]{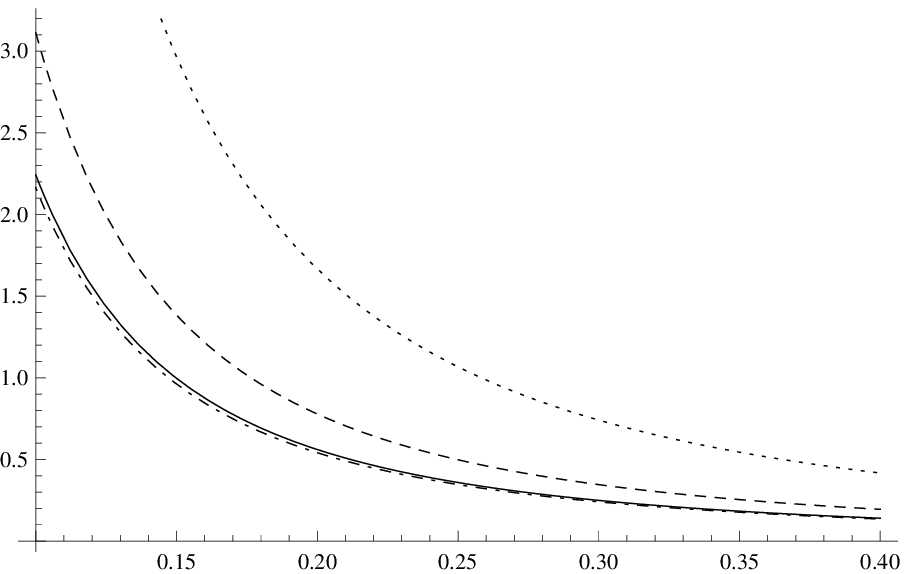}\\
$V$ vs $\sig$ & $\pi^*$ vs $\sig$\\
\includegraphics[width=0.475\textwidth]{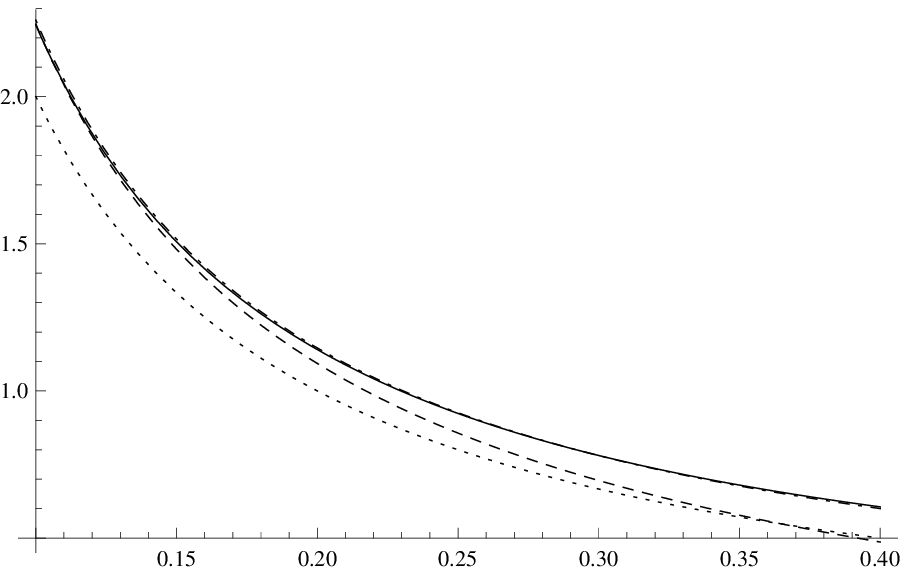}&
\includegraphics[width=0.475\textwidth]{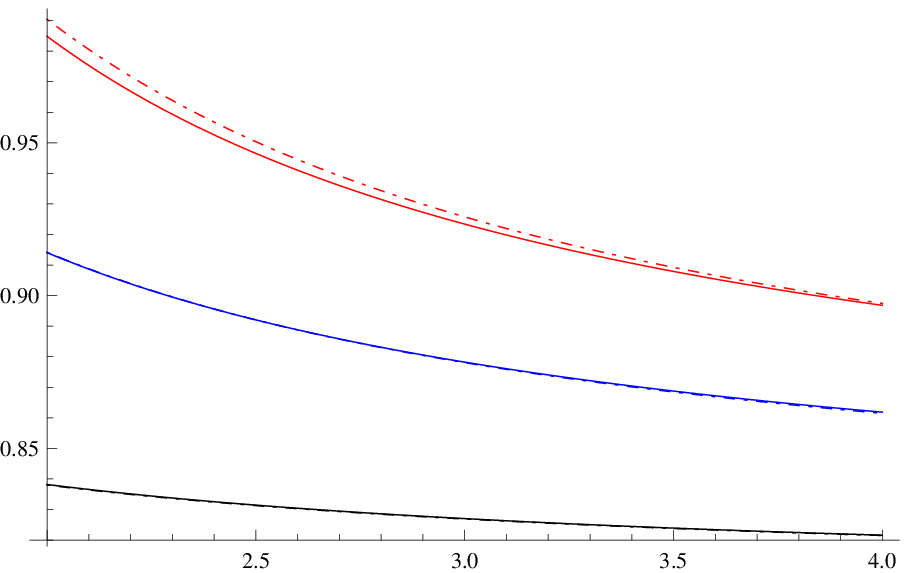}\\
$\Lam$ vs $\sig$ & $\Lam$ vs $\gam$
\end{tabular}
\caption{\small{{\em Value function (top left), optimal strategy (top right) and implied Sharpe ratio (bottom left) are plotted as a function of $\sig=\del \ee^{\eta x}$ assuming power utility \eqref{eq:CRRA} and dynamics given by local volatility model \eqref{eq:local.vol}.  
In all three plots, the solid line corresponds to the exact function, and the dotted, dashed and dot-dashed lines correspond to the zeroth, first and second-order approximations, respectively.
The parameters used in these three plots are: 
$T-t=5$, $w=1.0$, $\eta = -0.8$, $\del = 0.3$, $\mu = 0.3$ and $\gam =3.00$.
On the bottom right we fix the volatility $\sig=0.25$ and we plot the implied Sharpe ration $\Lam$ as a function of $\gam$ for three different time horizons $T-t=\{1,3,5\}$ corresponding to black, blue and red, respectively.  The solid lines are exact.  The dot-dashed lines are our second order approximation.
The parameters used in the bottom right plot are: $w=1.0$, $\eta = -0.8$, $\del = 0.3$ and $\mu = 0.3$.}}
}
\label{fig:cev-2}
\end{figure}

\section{Conclusion}
\label{sec:conclusion}
In this paper we consider the finite horizon utility maximization problem in a general LSV setting.  Using polynomial expansion methods, we obtain an approximate solution for the value function and optimal investment strategy.  The zeroth-order approximation of the value function and optimal investment strategy correspond to those obtained by \cite{merton1969lifetime} when the risky asset follows a geometric Brownian motion.  

The first-order correction of the value function can always be expressed as a differential operator acting on the zeroth-order term.  Higher-order corrections can always be expressed as a nonlinear transformation of a convolution with a Gaussian kernel.  For certain utility functions, these convolutions can be expressed in closed-form as a differential operator acting on the zeroth-order term.  Corrections to the zeroth-order optimal investment strategy can be obtained from the approximation of the value function.  

We also introduce in this paper the concept of an implied Sharpe ratio and derive an approximation for this quantity.  We obtain specific results for power utility and give a rigorous error bound for the value function in a stochastic volatility setting.  Finally, we provide two numerical examples to illustrate the accuracy and versatility of our approach.  The expansion techniques presented in this paper naturally lend themselves to other nonlinear stochastic control problems. Recent results for indifference pricing of options contracts have been developed in \cite{lorig-4}.

\bibliographystyle{chicago}
\bibliography{Bibtex-Master-3.02-submitted}

\end{document}